\documentclass{lmcs}
\pdfoutput=1

\usepackage{lastpage}
\lmcsdoi{17}{3}{13}
\lmcsheading{}{\pageref{LastPage}}{}{}%
{Mar.~04,~2019}{Jul.~28,~2021}{}

\usepackage{amssymb}
\usepackage{stmaryrd}

\usepackage[T1]{fontenc}
\usepackage[utf8]{inputenc}
\usepackage{newunicodechar}

\DeclareUnicodeCharacter{25C2}{\text{$\triangleleft$}}
\newunicodechar{⟨}{\langle}
\newunicodechar{⟩}{\rangle}
\newunicodechar{→}{\text{$\rightarrow$}}
\newunicodechar{λ}{\text{$\lambda$}}
\newunicodechar{₁}{\text{${}_1$}}
\newunicodechar{₂}{\text{${}_2$}}
\newunicodechar{⋆}{\text{$*$}}
\newunicodechar{Σ}{\text{$\Sigma$}}

\newunicodechar{≡}{\text{$\equiv$}}
\newunicodechar{≈}{\text{$\approx$}}
\newunicodechar{ν}{\text{$\nu$}}
\newunicodechar{₀}{\text{${}_0$}}

\usepackage{alltt}

\newenvironment{allttt}
  {\par\begingroup\footnotesize\begin{alltt}\ignorespaces}
  {\end{alltt}\endgroup\ignorespacesafterend\noindent}

  \usepackage{tikz}
  \usetikzlibrary{matrix,arrows,calc,positioning}

  \tikzset{morphism/.style={->,font=\scriptsize}}
  \tikzset{injection/.style={right hook->,font=\scriptsize}}
  \tikzset{twocell/.style = {double equal sign distance,double,-implies,shorten <= .3cm,shorten >=.3cm,font=\scriptsize,draw}}
  \tikzset{crossover/.style={preaction={solid,draw=white,-,line width=6pt}}}
  \tikzset{over/.style={fill=white,inner sep=1.5pt}}

\tikzset{over/.style={fill=white,inner sep=2.5pt}}

\newcommand{\B}{\ensuremath{\mathbb{B}}}
\newcommand{\N}{\ensuremath{\mathbb{N}}}

\newcommand{\C}{\ensuremath{\mathbb{C}}}

\newcommand\Set{\ensuremath{\mathsf{Set}}}

\newcommand{\Fam}{\ensuremath{\mathsf{Fam}}}
\newcommand{\Pow}{\ensuremath{\mathsf{Pow}}}

\newcommand{\End}{\ensuremath{\mathsf{End}}}
\newcommand{\Mnd}{\ensuremath{\mathsf{Mnd}}}

\newcommand{\Nat}{\ensuremath{\mathsf{Nat}}}
\newcommand\Alg[1]{\ensuremath{#1{\mathtt{-}}\mathsf{alg}}}
\newcommand\MAlg[1]{\ensuremath{#1{\mathtt{-}}\mathsf{Alg}}}

\newcommand\Rule[2]{\frac{\phantom{\big(}\quad{\displaystyle #1}\quad}{\phantom{\Big(}\quad{\displaystyle #2}\quad}}

\newenvironment{myequation}{%
  \ignorespaces\[\everymath{\displaystyle}\begin{array}{rclr}%
    }{%
  \end{array}\]\ignorespacesafterend%
}
\newenvironment{verticalhack}
  {\begin{array}[b]{@{}c@{}}\displaystyle}
  {\\[-0.4em]\noalign{\hrule height0pt}\end{array}}
\newenvironment{myequationqed}{%
  \ignorespaces\[\begin{verticalhack}\everymath{\displaystyle}\begin{array}{rclr}%
    }{%
  \end{array}\end{verticalhack}\qedhere\]\ignorespacesafterend%
}

\newcommand\BLANK{{\texttt{\char"5F}}}
\newcommand{\restrict}{\upharpoonright}
\renewcommand\t[1]{\leavevmode\text{\textup{\texttt{#1}}}}

\newcommand{\ie}{i.e.\ }

\newcommand\sub{\subset}
\newcommand\meets{\between}
\newcommand\IN{\mathrel{\epsilon}}

\renewcommand{\implies}{\Rightarrow}
\newcommand{\isimpliedby}{\Leftarrow}

\newcommand{\NU}[1]{\nu_{#1}}
\newcommand{\LAM}[1]{(\lambda\,#1)\,{}}
\newcommand\nuIntro{\nu_\mathrm{intro}}
\newcommand\nuElim{\nu_\mathrm{elim}}
\newcommand\Sem[1]{\left\llbracket#1\right\rrbracket}
\newcommand\linear{\multimap}

\newcommand{\PiSymbol}{\Pi}
\newcommand{\PI}[1]{\left(\big.\PiSymbol\,#1\right)\,{}}
\newcommand{\SigmaSymbol}{\Sigma}
\newcommand{\SI}[1]{\left(\big.\SigmaSymbol\,#1\right)\,{}}

\newcommand{\id}{\mathrm{id}}
\newcommand{\eval}{\mathrm{eval}}

\newcommand\ALayered[1]{\langle\mskip-4.5mu\langle{#1}\rangle\mskip-4.5mu\rangle}

\newcommand\Stream[1]{\mathtt{stream}(#1)}
\newcommand\Tree[2]{\mathtt{tree}_{#1}(#2)}

\newcommand{\One}{{\bf 1}}
\newcommand{\Zero}{{\bf 0}}
\newcommand{\Two}{{\bf 2}}

\newcommand{\OUT}{\mathtt{output}}
\newcommand{\INPUT}{\mathtt{input}}
\newcommand{\tINPUT}{\t{input}}
\newcommand{\leaf}{\mathtt{Leaf}}
\newcommand{\node}{\mathtt{Node}}

\usepackage{pifont}
\newcommand{\checked}{{\bf{[}{\it Agda}\/\ding{51}]}\enspace}%
\newcommand{\partlychecked}{{\bf{[}{\it Agda}\/\ding{55}]}\enspace}%

\begin{document}


\title[Continuous Function, Greatest Fixed Points and Indexed Containers]
      {Representing Continuous Functions between Greatest Fixed Points of Indexed Containers}


\author{Pierre Hyvernat}
\email{\tt pierre.hyvernat@univ-smb.fr}
\urladdr{\url{https://www.lama.univ-savoie.fr/~hyvernat/}}
\address{
Universit\'e Grenoble Alpes, Universit\'e Savoie Mont Blanc, CNRS, LAMA, 73000
Chamb\'ery, France.}

\thanks{The author wants to thank {both} Peter Hancock and the reviewers. The
former should get credits for many of the ideas developed in this paper. He
visited Chamb\'ery many years ago and most of the results started to take form
then and there. The latter contributed many helpful comments much more
recently when the paper was finally submitted.}

\begin{abstract}
  We describe a way to represent computable functions between coinductive
  types as particular transducers in type theory. This generalizes earlier
  work on functions between streams by P. Hancock to a much richer class of
  coinductive types. Those transducers can be defined in dependent type theory
  \emph{without} any notion of equality but require inductive-recursive
  definitions. Most of the properties of these constructions only rely on a
  mild notion of equality (intensional equality) and can thus be formalized in
  the dependently typed language Agda.
\end{abstract}

\maketitle


\section*{Introduction}

This paper gives a type theoretic representation theorem for continuous
functions (Theorems~\ref{lem:eval_layer} and~\ref{lem:eval_layer_comp}) between
a wide class of spaces of infinite values. By infinite, we understand a value
having a coinductive type. Continuity means that finite information about a
result of the function requires only finite information about its argument.
Because of that, it is a necessary condition for computability.

The simplest coinductive space is the Cantor space of infinite boolean
sequences. It corresponds to the coinductive type~$\nu_X(X + X)$ and its
topology is well known. Programs that implement continuous functions from the
Cantor space to a discrete space~$D$ can be represented by finite binary trees
in which leaves are labelled by values in~$D$. We extend this to a class of
functors going beyond~$X \mapsto X + X$ on the category~$\Set$ by considering
so-called (finitary) polynomial functors on~$\Set^I$ for some index set~$I$.
The final coalgebra~$\nu P$ of such a functor~$P$ always exists and may be
constructed as the inverse limit of:~$1 \leftarrow P(\One) \leftarrow
P^2(\One) \leftarrow \cdots$. Those final coalgebras have a natural topology,
and when the functor~$F$ is \emph{finitary} (commutes with filtered colimits),
the topology enjoys a close connection with the intuitive notion of ``finite
amount of information'' about potentially infinite values.  However,
representing such topologies inside a formalized system such as dependent type
theory is far from trivial because their definition relies heavily on delicate
topics like equality.

Our main result pertains to the question of how continuous functions between
these natural classes of spaces can be represented in dependent type theory. It
turns out that any such ``implementation'' can itself be put into the form of a
potentially infinite data-structure, inhabiting a final coalgebra for an
appropriate functor, albeit one which is in most cases no longer finitary. This
settles a conjecture of P. Hancock about representability of continuous between
``dependent streams''~\cite{hancock09:_repres_of_stream_proces_using} by
extending the work of P. Hancock, N. Ghani and D.
Pattinson~\cite{hancock09:_contin_funct_final_coalg} from \emph{containers} to
\emph{indexed containers}. The technology and root ideas are very similar.

We obtain this result via a more general construction, without any cardinality
restrictions on the initial functors. One can still topologise the final
coalgebras, though the topology that arises from the inverse chain
construction no longer enjoys much connection with any intuition of finite
information, and there are (classically) continuous functions that cannot be
implemented by programs.


\section{Preliminaries I. Streams and Trees in Point Set Topology}

\subsection{Streams}  

Given a set~$X$ endowed with the discrete topology, the set of \emph{streams}
over~$X$, written~$\Stream{X}$, is defined as the infinite
product~$\prod_{i\geq0} X$. The product topology is generated from the basic
open sets~$\prod_{i\geq0} U_i$ where finitely many~$U_i$s are of the
form~$\{x_i\}$ for some~$x_i \in  X$ and the other~$U_i$s are equal to~$X$.
This topological space is usually called the Cantor space (when~$X$ is finite)
or the Baire space (when~$X$ is countably infinite). Continuity for functions
between streams amounts to the following:

\begin{lem}\label{lem:continuous_stream}
  A function~$f : \Stream{X} \to \Stream{Y}$ is continuous if and only if, for
  each stream~$s$ in~$\Stream{X}$, each projection~$f(s)_k$ of~$f(s)$ depends
  on at most
  a finite prefix of~$s$.
\end{lem}
Writing~$s_{\restrict n}$ for the restriction of stream~$s$ to its finite
prefix of length~$n$, the condition is equivalent to
\begin{equation}\tag{\ensuremath{\ast}}\label{eqn:prefix_condition}
  \forall s\in \Stream{X},
  \forall k\geq0,
  \exists n\geq0,
  \forall t\in \Stream{X},
  s_{\restrict n} = t_{\restrict n} \implies f(s)_k = f(t)_k
  \,.
\end{equation}
Before proving Lemma~\ref{lem:continuous_stream}, let's look at a preliminary
result.
\begin{lem}\label{lem:stream_openset}
  For any subset~$V\subseteq \Stream{X}$, we have:~$V$ is open iff
  \[
    \forall s\in V, \exists n\geq 0, \forall t \in \Stream{X}, s_{\restrict n}
    = t_{\restrict n} \implies t \in V
    \,.
  \]
\end{lem}

\begin{proof}
  The $\implies$ direction is immediate: an open set is a union of basic open
  sets, which satisfy the condition. (Recall that a basic open set is of the
  form~$\prod_{i\geq0} U_i$, where each~$U_i$ is~$X$, except for finitely many
  that are singleton sets.)

  For the~$\isimpliedby$ direction, we define, for each~$s\in V$, the set~$V_s
  = \{t \mid s_{\restrict n_s}=t_{\restrict n_s}\}$, where~$n_s$
  is the integer coming from condition. We have~$V = \bigcup_{s\in V} V_s$.
\end{proof}

\begin{proof}[Proof of Lemma~\ref{lem:continuous_stream}]
  Suppose the function~$f : \Stream{X} \to \Stream{Y}$ satisfies
  condition~$(\ref{eqn:prefix_condition})$. To show that~$f$ is continuous,
  it is enough to show that the inverse image of any basic open set is an open
  set. Because the inverse image commutes with intersections, it is sufficient
  to look at \emph{pre} basic open sets of the form~$V_{k,y}=\{s \mid s_k =
  y\}$.

  To show that~$f^{-1}(V_{k,y})$ is open, we use
  Lemma~\ref{lem:stream_openset} and show that~$s\in
  f^{-1}(V_{k,y})$ implies
  \[
    \forall s \in f^{-1}(V_{k,y})
    \exists n\geq 0, \forall t \in \Stream{X}, s_{\restrict n}
    = t_{\restrict n} \implies t \in f^{-1}(V_{k,y})
    \,.
  \]
  Because~$s\in f^{-1}(V_{k,y})$ is~$f(s)_k = y$, this is implied by
  condition~$(\ref{eqn:prefix_condition})$.

  For the converse, suppose~$f : \Stream{X} \to \Stream{Y}$ is continuous. We
  want to show that it satisfies condition~$(\ref{eqn:prefix_condition})$.
  Let~$s\in \Stream{X}$ and~$k\geq0$. The set~$\{t \mid f(t)_k = f(s)_k\}$ is
  open and because~$f$ is continuous, its inverse image also is open. By
  Lemma~\ref{lem:stream_openset}, we now that there is some~$n$ such that
  $s_{\restrict n} = t_{\restrict n} \implies f(t)_k = f(s)_k$. This finishes
  the proof.
\end{proof}
Because of this, \emph{constructive} functions between streams
are usually held to be continuous. We expect them to arise as continuous
functions with the additional properties that:
\begin{itemize}
  \item
    finding the finite prefix needed to compute a chosen element of~$f(s)$ is
    computable, and

  \item
    finding the value of the element of~$f(s)$ from that finite prefix is
    computable.
\end{itemize}

Note that the discrete space~$X$ may be generalized to a
family~$(X_i)_{i\geq0}$ that need not be constant. More interestingly, we can
allow the set~$X_i$ (giving the set of possible values for the~$i$th element
of a stream) to depend on the~$i$th prefix of the stream. We can in this way
obtain the space of increasing streams of natural numbers:

\begin{itemize}
  \item
    the set~$X_0$ does not depend on anything and is defined as~$\N$,

  \item
    the set~$X_1$ depends on the value of~$x_0\in X_0$: $X_{1,x_0} = \{ k\in
    \N , x_0 \le k \}$,

  \item
    the set~$X_2$ depends on~$x_1\in X_{1,x_0}$, etc.
\end{itemize}
The set of increasing streams is not naturally a product space but is a
subspace of~$\Stream{\N}$. Because of that, the topology is the expected one and
continuous functions are still characterized by
Lemma~\ref{lem:continuous_stream}.


\subsection{Infinite Trees, Natural Topology}  
\label{sub:natural}

The natural topology for sets of \emph{infinite trees} is less well known than
the Cantor and Baire topologies. The simplest kind of infinite tree, the
infinite \emph{binary tree} has a root, and two distinguished ``branches''
going from that root to two ``nodes''. Each of these two nodes also has two
branches, etc. An infinite binary tree \emph{over}~$X$ is a way to label each
node of the infinite binary tree with an element of~$X$.
If we write~$\B$ for the set~$\{0,1\}$, each node of the infinite binary tree
can be identified by a list of elements of~$\B$: this list simply represents the
branch leading to this node from the root. The set of infinite binary trees
over~$X$, written~$\Tree{\B}{X}$, can thus be defined as
\[
  \Tree{\B}{X} = X \times (\B\to X)
                  \times (\B^2\to X)
                  \times \dots
                  \times (\B^i\to X)
                  \times \dots
\]
where each term gives the~$i$th ``layer'' of the tree as a function from finite
branches of length~$i$ to~$X$. We can rewrite this as
\[
  \Tree{\B}{X}
  \quad=\quad
  \prod_{i\geq0} \big(\B^i \to X\big)
  \quad=\quad
  \prod_{i\geq0} \left(\prod_{t\in \B^i} X\right)
  \,.
\]
By replacing the set~$\B$ by some other set~$B$, we obtain the ternary trees
over~$X$ or countably-branching trees over~$X$, etc. Streams themselves are
recovered by taking~$B = \{\star\}$.
If both~$B$ and~$X$ are endowed with the discrete topology, we obtain a
natural topology on~$\Tree{B}{X}$. Note that when~$B$ is infinite, the
spaces~$B^i \to X$ are not discrete anymore. Nevertheless, we have:

\begin{lem}\label{lem:continuous_trees}
  Let~$A$, $B$, $X$ and~$Y$ be discrete spaces; a function~$f : \Tree{A}{X}
  \to \Tree{B}{Y}$ is continuous iff for every~$t\in \Tree{A}{X}$, the value
  at each node of~$f(t)$ only depends on a finite subtree\footnote{A
  \emph{subtree} is a set of nodes that contains the root of the tree and is
  closed by the ancestor relation.} of~$t$.
\end{lem}

\begin{proof}
  The proof of this lemma is exactly the same as the proof of
  Lemma~\ref{lem:continuous_stream}, except that we replace the natural
  number~$n$ in~$t_{\restrict n}$ (for some~$t\in \Tree{B}{Y}$) by a
  \emph{finite subtree}.
  The only remark is that basic open sets of~$\Tree{B}{Y}$ are of the form
  $\prod_{i\geq0} \prod_{t\in B^i} X_{i,t}$ where all sets~$X_{i,t}$ are equal
  to~$X$, except for finitely many that are singletons of the
  form~$\{x_{i,t}\}$ for some~$x_{i,t} \in X$.
\end{proof}
It is again possible to devise more general notions of trees by allowing the
set~$X$ at a node in the tree to depend on the values of its location as a
path from the root. The resulting space is endowed with the subspace topology
and Lemma~\ref{lem:continuous_trees} still holds.
We will later generalize this notion further by allowing the branching of a
node (given by the set~$B$) to depend itself on the value stored at the node.
With this generalisation we can model very general objects, such as infinite
automata that issue commands and change to a new state (choose a branch) based
on the responses.


\subsection{Infinite Trees, Wild Topology}  
\label{sub:wild}

The topology we naturally get in this paper corresponds to a
different topology on trees. When looking at
\[
  \Tree{B}{X} \quad=\quad \prod_{i\geq0} \left(\prod_{t\in B^i} X\right)
  \,,
\]
we can endow the inner product space with the ``box'' topology, where basic
opens are given by arbitrary products of open sets. Because~$B$ and~$X$ are
discrete sets, this amounts to giving the discrete topology to each
layers~$B^i \to X$. Instead of being generated by ``finite subtrees'', open
sets are generated by ``subtrees with finite depth''.

\begin{lem}\label{lem:continuous_trees_wild}
  Let~$A$, $B$, $X$ and~$Y$ be discrete spaces; a function~$f : \Tree{A}{X}
  \to \Tree{B}{Y}$ is continuous \emph{for the wild topology} iff for
  every~$t\in \Tree{A}{X}$, and~$k\in\N$, there is an~$n\in\N$ such that
  the nodes of~$f(t)$ at depth less than~$k$ depend only on the nodes of~$t$
  at depth less than~$n$.
\end{lem}
Formally, this looks very similar to condition~$(\ref{eqn:prefix_condition})$ on
page~\pageref{eqn:prefix_condition}:
\begin{equation}\tag{\ensuremath{\ast\ast}}
  \forall t\in \Tree{A}{X},
  \forall k\geq0,
  \exists n\geq0,
  \forall t'\in \Tree{A}{X},
  t_{\restrict k} = t'_{\restrict k} \implies f(t)_n = f(t')_n
  \,,
\end{equation}
where~$t_{\restrict k}$ is the complete subtree of~$t$ up-to depth~$k$.
Intuitively, this topology considers infinite trees as streams of their
layers, where layers are discrete.

\medbreak
When~$A$ and~$B$ are finite, the two notions of continuity
(Lemma~\ref{lem:continuous_trees} and~\ref{lem:continuous_trees_wild})
coincide. When this is not the case, we cannot compare continuous
functions for the two topologies.

\begin{itemize}

  \item\label{ex:natural_not_wild}
    Consider~$f:\Stream{\N} \to \Tree{\N}{\N}$ sending the stream~$s$ to the tree~$f(s)$
    where the node indexed by~$(i_0,\dots,i_n)$
    is~$s_{i_0}+s_{i_1}+...+s_{i_n}$. This is certainly continuous for the
    natural topology. However, because the first layer of the output is
    infinite, we cannot bound the number of layers (elements) of the input
    stream~$s$ that are needed to construct it: this function is not continuous
    for the wild topology.

  \item
    Consider~$g:\Tree{\N}{\B} \to \Stream{\B}$ where the~$i$th element
    of~$g(t)$ is the maximum of the complete~$i$th layer of~$t$ ($\B$ being
    the complete lattice with two elements). This function is continuous in
    the wild sense, but because we need to know the whole~$i$th layer of the
    input to get the~$i$th value of the output, this function is not continuous
    for the natural topology (and certainly not computable).

\end{itemize}



\section{Preliminaries II. Martin L\"of Type Theory}

\subsection{Basic Features}  

We work in a meta theory that is in essence Martin-L\"o{}f's dependent type
theory~\cite{ML73} with two additional features:

\begin{itemize}

  \item
    coinductive types,

  \item
    inductive-recursive definitions.

\end{itemize}
All the constructions described in the paper have been defined using the
dependently typed functional programming language Agda~\cite{Agda}.

\subsubsection*{A Note about Equality}     

This paper is concerned with constructions in ``pure'' type theory, \ie
dependent type theory without identity. Those constructions enjoy many
interesting properties, but proving them requires some notion of equality.
Equality in Martin-L\"of type theory is a complex subject about which whole
books have been written~\cite{HOTTbook}. We try to be mostly agnostic about
the flavor of equality we are using and only rely on the ``simplest'' one:
intentional equality, written~$a \equiv_T b$, or simply~$a \equiv b$. This
makes it possible to use vanilla Agda for checking proofs.\footnote{All the
Agda code was checked using \t{Agda 2.6.1.1} with the flag \t{--without-K}.
The Agda code is available
at~\url{http://www.lama.univ-smb.fr/~hyvernat/Files/Infinite/agda.tgz}, with
the file \t{PAPER.agda} referencing all the formalized results from the paper.
For those without a working Agda installation, the code is also browsable
directly
from~\url{http://www.lama.univ-smb.fr/~hyvernat/Files/Infinite/browse/PAPER.html}.
}

We annotate the proofs in
the paper with
\begin{itemize}
  \item
    \checked: for those that have been formalized using Agda,

  \item
    \partlychecked: for those which have only partly been formalized,
    typically because the proof is too complex to write in Agda, or because it
    requires a much stronger version of equality than we have at our disposal.
\end{itemize}
Because we want to explain Agda code only slightly less than you want to read
Agda code, the formalized proofs are either omitted from the paper, or
explained informally.

We sometimes need to assume equality for functions is extensional, \ie that~$f
\equiv_{} g$ iff~$f\,a \equiv_B g\,a$ for all~$a: A$. Those proofs are clearly
identified.


\subsubsection*{Notation}  

The notation for dependent types is standard. Here is a summary:
\begin{itemize}

  \item We write~$\Gamma \vdash B$ to mean that~$B$ is a well-formed type in
    context~$\Gamma$.  We write $A = B$ to express that $A$ and $B$ are
    the same by definition.

  \item We write~$\Gamma \vdash b:B$ to mean that~$b$ is an element of
    type~$B$ in context~$\Gamma$. The context is often left implicit and we
    usually write~$b:B$.
    We write~$a = b : B$ to express that~$a$ and~$b$ are
    definitionally equal in type~$B$. When the type $B$ can easily be deduced
    from the context, we will usually write just~$a = b$.

  \item If $B$ is a type and~$C$ is a type depending on~$x : B$,
    \ie,~$x:B\vdash C$, we write
    \begin{itemize}
      \item $\SI{x:B}C$ for the dependent sum. Its canonical elements are pairs~$⟨
        b, c⟩$ with $b : B$ and $c : C[x/b]$,
      \item $\PI{x:B}C$ for the dependent product. Its canonical elements are
        functions~$\LAM{x:B}u$ where $x:B \vdash u : C$.
    \end{itemize}
    When the type~$C$ is constant, we abbreviate those by~$B \times C$ and~$B \to C$.

  \item The usual amenities are present: the natural numbers, W-types and so on.
\end{itemize}
We use a universe~$\Set$ of ``small types'' containing~$\Zero$ (with no
element),~$\One$ (with a single element~$\star$) and~$\Two$ (with two
elements~$0$ and~$1$). Moreover, the dependent sums and products are reflected
in this universe, and we use the same notation~$\SI{b:B}C$ and~$\PI{b:B}C$
whenever~$B:\Set$ and~$b:B \vdash C:\Set$.
We assume that this universe is closed under many inductive-recursive and
coinductive definitions which will be treated below.

\smallbreak
We are not always consistent with notation for
application of functions and usually write
\begin{itemize}
  \item
     $f\,x$ when the result is an element of a small type (an element
     of~\Set),
  \item
    $A(i)$ when the result is itself a small type (and thus, not an element
    of~\Set).
\end{itemize}


\subsubsection*{Predicates and Families}  

The Curry-Howard isomorphism makes the type~$\Set$ into a universe of
propositions.
\begin{defi}
  If~$A:\Set$, the collection of \emph{predicates} on~$A$
  is defined as
  \[
    \Pow(A) \quad = \quad A \to \Set
    \,.
  \]
  We introduce the following notations~\cite{sambin93:_toolbox}: if~$X$
  and~$Y$ are predicates on~$A$,
  \begin{itemize}
    \item ``$a \IN X$'' is~``$X(a)$'', 

    \item ``$X \sub Y$'' is~``$\PI{a:A} (a\IN X) \to (a\IN Y)$,

    \item ``$X \meets Y$'' is~``$\SI{a:A} (a\IN X) \times (a\IN Y)$'',

    \item ``$X \cap Y$'' is~``$\LAM{a:A}(a \IN X)\times(a\in Y)$'',

    \item ``$X \cup Y$'' is~``$\LAM{a:A}(a \IN X)+(a\in Y)$''.

  \end{itemize}
\end{defi}
The intuition is that a predicate on~$A$ is just a subset of~$A$ in some
constructive and predicative\footnote{if~$A:\Set$, $\Pow(A)$ is not of
type~$\Set$} set theory. It is then natural to call a predicate on
some~$A\times B$ a binary \emph{relation}. We sometimes identify a relation
$R:\Pow(A\times B)$ with its curried version~$R : A \to \Pow(B)$.
Finally, we define the following trivial operation
\begin{defi}
  If $R : \Pow(A\times B)$, the converse $R^\sim$ is the predicate in~$\Pow(B\times A)$ given by
  \[
  R^\sim(b,a) \quad=\quad R(a,b)
  \ .
  \]
\end{defi}

We will also need a notion of family.
\begin{defi}
  A family of~$C$ is given by a set~$I$ together with a function from~$I$
  to~$C$. In other
  words,
  \[
    \Fam(C)
    \quad = \quad
    \SI{I:\Set} I \to C
    \,.
  \]
\end{defi}


\subsection{Inductive-recursive Definitions}  

Inductive definitions are a way to define (weak) initial algebras for
endofunctors on~$\Set$. A typical example is defining~$\mathtt{list}(X)$ as
the least fixed point of~$Y\mapsto \One + X \times Y$.
In their simplest form, inductive-recursive definitions~\cite{dybjer:JSL2000}
give a way to define (weak) initial algebras for endofunctors
on~$\Fam(C)$.\footnote{To ensure the definition is healthy, the endofunctor
has to be expressible according to a certain coding
scheme~\cite{dybjersetzer:APAL2003}.} It means we can define, at the same
time:
\begin{itemize}
  \item a inductive set~$U$, called the \emph{index set},
  \item and a recursive function~$f : U \to C$.
\end{itemize}
Of course, the set~$U$ and the function~$f$ may be mutually dependent.

\smallbreak
The type~$C$ may be large or small. Taking~$C=\One$, we recover usual
inductive types as the index part of such a definition, the recursive
function~$f$ always being trivial. In non-degenerate cases, the inductive
clauses defining~$U$ are expressed using the simultaneously defined
function~$f$. Here is a traditional example with~$C=\Set$: complete binary
trees indexed by~$X$, defined one layer at a time. Here is the definition,
first using Agda syntax:
\begin{allttt}\label{ex:TreeIR}
  mutual
    data Tree (X : Set) : Set where
      Empty : Tree X
      AddLayer : (t : Tree X) → (branches t \(\times\) Bool → X) → Tree X
    branches : \{X : Set\} → Tree X → Set
    branches Empty = One
    branches (AddLayer t l) = (branches t) \(\times\) Bool
\end{allttt}
While~\t{Empty} corresponds to the empty tree, the
definitions
\begin{allttt}
    T₁ = AddLayer Empty (λ b → if b.2 then 2 else 1)
    T₂ = AddLayer T₁ (λ b →
                        if   b.1.2 \&\&     b.2 then 6
                        elif b.1.2 \&\& not b.2 then 5
                        elif not b.1.2 \&\& b.2 then 4
                        else                       3)
\end{allttt}
give trees labeled by natural numbers: (putting the ``{\t{False}'' branch on the left, and the
``\t{True}'' branch on the right})
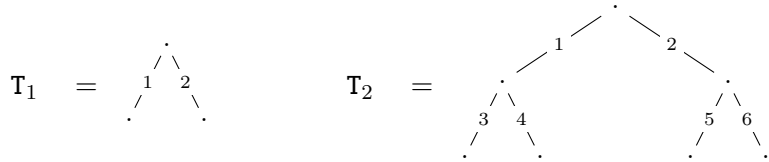
\[
  \mathtt{T}_1 \quad = \ %
  \begin{tikzpicture}[
        level distance=1cm,
        level 1/.style={sibling distance=1cm},
        tree node/.style={draw=none},
        every child node/.style={tree node},
        baseline={([yshift=-1ex] current bounding box.center)},
        ]
    \node[tree node] (Root) {.}
      child {node {.} edge from parent node[over] {\tiny1}}
      child {node {.} edge from parent node[over] {\tiny2}};
  \end{tikzpicture}
  \mskip80mu
  \mathtt{T}_2 \quad = \ %
  \begin{tikzpicture}[
        level distance=1cm,
        level 1/.style={sibling distance=3cm},
        level 2/.style={sibling distance=1cm},
        tree node/.style={draw=none},
        every child node/.style={tree node},
        baseline={([yshift=-1ex] current bounding box.center)},
        ]
    \node[tree node] (Root) {.}
    child {
      node {.}
        child {node {.} edge from parent node[over] {\tiny 3}}
        child {node {.} edge from parent node[over] {\tiny 4}}
      edge from parent node[over] {\tiny1}}
    child {
      node {.}
        child {node {.} edge from parent node[over] {\tiny 5}}
        child {node {.} edge from parent node[over] {\tiny 6}}
      edge from parent node[over] {\tiny2}};
  \end{tikzpicture}\]
The corresponding functor takes the family~$⟨ T:\Set, b:T\to\Set⟩$
to the family
\begin{itemize}
  \item
    index set: $T' = \One + \SI{t:T} \big((b\,t \times \Two) \to \N\big)$,

  \item
    recursive function~$b'$ defined with
    \begin{itemize}
      \item $b'\,\star = \One$, where~$\star$ is the only element of~$\One$,
        \item $b'\,⟨ t, l⟩ = (b\,t)\times\Two$.
    \end{itemize}

\end{itemize}
We will, in Section~\ref{sub:layering}, encounter a similar situation in which
the type~$C$ will be~$\Fam(I)$, \ie we will need to take the least fixed point
of a functor from~$\Fam\big(\Fam(I)\big)$ to itself.

Another typical example involves defining a universe~$U$ of types closed under
dependent function space:~$U$ needs to contains inductive elements of the
form~$\PI{A:U}(B:\mathsf{fam}(A))$, but~$\mathsf{fam}(A)$ is defined
as~$\mathsf{El}(A) \to U$ and makes use of the decoding function~$\mathsf{El}
: U \to \Set$.


\subsection{Greatest Fixed Points}  

In its simplest form a coinductive definition introduces some~$\NU{F}:\Set$ together
with a weakly terminal coalgebra~$c : \NU{F} \to F(\NU{F})$ for a sufficiently
healthy\footnote{for our purposes, ``sufficiently healthy'' amounts to
``polynomial''} functor~$F : \Set \to \Set$. For example, given a set~$A$, the
functor~$ F(X) = A \times X$ is certainly healthy and the set~$\NU{F}$
corresponds to the set of streams over~$A$. The coalgebra~$c$ ``decomposes'' a
stream~$[a_0, a_1, \ldots]$ into the pair~$⟨ a_0 , [a_1,a_2,\ldots]
⟩ $ of its head and its tail. Because the coalgebra is weakly terminal,
for any other such coalgebra~$x : X \to F(X)$ there is a map $X \to
\NU{F}$.
%
The corresponding typing rules are given by
\[
  \Rule{\sigma : X \to F(X)}{\nuIntro\,\sigma : X \to \NU{F}}
  \qquad\text{and}\qquad
  \Rule{}{\nuElim : \NU{F} \to F(\NU{F})}
\]
and the computation rule is
\[
  \nuElim\,(\nuIntro\,\sigma\,x) \quad = \quad F_{\nuIntro\,\sigma} \, (\sigma\,x)
  \,.
\]
Such coinductive definitions can be extended to families of sets: given an
index set~$I$, we introduce weakly terminal coalgebras for a sufficiently
healthy functor acting on~$\Pow(I)$.
The typing rules are extended as follows:
\[
  \Rule{\sigma : X \sub F(X)}{\nuIntro\,\sigma : X \sub \NU{F}}
  \qquad\text{and}\qquad
  \Rule{}{\nuElim : \NU{F} \sub F(\NU{F})}
\]
together with the computation rule
\[
  \nuElim\,(\nuIntro\,\sigma\,i\,x)
  \quad = \quad
  F_{\nuIntro\,\sigma}\,i \,  (\sigma\,i\,x)
  \,.
\]
Note that it does not seem possible to guarantee the existence of strict
terminal coalgebras~\cite{LetsUnfold} without extending considerably the type
theory.\footnote{It is apparently possible to do so in univalent type
theory~\cite{Ahrens}, where coinductive types can be defined from inductive
ones, just like in plain old set theory!}

\subsubsection*{Bisimulations}  

Equality in type theory is a delicate subject, and it is even more so in the
presence of coinductive types. The usual (extensional or intensional) equality
is easily defined and shown to enjoy most of the expected properties. However,
it is not powerful enough to deal with infinite objects. Two infinite objects
are usually considered ``equal'' when they are \emph{bisimilar}. Semantically,
bisimilarity has a simple definition~\cite{Aczel,Staton,Ahrens}. In the
internal language, two coalgebras are bisimilar if their decompositions are
equal, coinductively.
\begin{defi}\label{def:categorical_bisimulation}
  Given a locally cartesian closed category~$\C$ with an endofunctor~$F$ and
  two coalgebra~$c_i : T_i \to F(T_i)$ ($i=1,2$), a \emph{bisimulation}
  between~$T_1$ and~$T_2$ is given by a span~$T_1 \leftarrow R \rightarrow
  T_2$ with a coalgebra structure such that the following diagram commutes.
  \[
    \begin{tikzpicture}
      \matrix (m) [matrix of math nodes, column sep={6em,between origins}, row
      sep={4em,between origins},text height=1.5ex, text depth=.25ex]
      {
        T_1    & R    & T_2    \\
        F(T_1) & F(R) & F(T_2) \\
      };
      \path[morphism]
      (m-1-2) edge node[above]{$r_1$} (m-1-1)
      (m-1-2) edge node[above]{$r_2$} (m-1-3)
      (m-2-2) edge node[below]{$F_{r_1}$} (m-2-1)
      (m-2-2) edge node[below]{$F_{r_2}$} (m-2-3)
      (m-1-1) edge node[left]{$c_1$} (m-2-1)
      (m-1-2) edge node[right]{$r$} (m-2-2)
      (m-1-3) edge node[right]{$c_2$} (m-2-3)
      ;
    \end{tikzpicture}
  \]
\end{defi}
In particular, the identity span~$T \leftarrow T \rightarrow T$ is always a
bisimulation and the converse of a bisimulation between~$T_1$ and~$T_2$ is a
bisimulation between~$T_2$ and~$T_1$.

Functions between coinductive types ought to be congruences for bisimilarity:
\begin{defi}
  If $f_i : C_i \to D_i$ ($i=1,2$) are morphisms from coalgebras $c_i:C_i\to
  F(C_i)$ to coalgebras~$d_i : D_i \to F(D_i)$, we say that~$f_1$ and~$f_2$ are
  equal \emph{up to bisimulation}, written~$f_1 \approx f_2$, if for every
  bisimulation~$C_1 \leftarrow R \rightarrow C_2$, there is a bisimulation~$D_1
  \leftarrow S \rightarrow D_2$ and a morphism~$h:R\to S$ making the following
  diagram commute.
  \[
    \begin{tikzpicture}
      \matrix (m) [matrix of math nodes, column sep={6em,between origins}, row
      sep={4em,between origins},text height=1.5ex, text depth=.25ex]
      {
        F(C_1) & F(R) & F(C_2) \\
        C_1  & R    & C_2  \\
        D_1  & S    & D_2  \\
        F(D_1) & F(S) & F(D_2) \\
      };
      \path[morphism]
      (m-1-2) edge node[above]{$F_{r_1}$} (m-1-1)
      (m-1-2) edge node[above]{$F_{r_2}$} (m-1-3)
      (m-2-1) edge node[left]{$c_1$} (m-1-1)
      (m-2-2) edge node[right]{$r$} (m-1-2)
      (m-2-3) edge node[right]{$c_1$} (m-1-3)
      (m-2-2) edge node[above]{$r_1$} (m-2-1)
      (m-2-2) edge node[above]{$r_2$} (m-2-3)
      (m-2-1) edge node[left]{$f_1$} (m-3-1)
      (m-2-3) edge node[right]{$f_2$} (m-3-3)
      (m-3-1) edge node[left]{$d_1$} (m-4-1)
      (m-3-3) edge node[right]{$d_2$} (m-4-3)
      ;
      \path[morphism,densely dashed]
      (m-2-2) edge node[right]{$h$} (m-3-2)
      (m-3-2) edge node[below]{$s_1$} (m-3-1)
      (m-3-2) edge node[below]{$s_2$} (m-3-3)
      (m-3-2) edge node[right]{$s$} (m-4-2)
      (m-4-2) edge node[below]{$F_{s_1}$} (m-4-1)
      (m-4-2) edge node[below]{$F_{s_2}$} (m-4-3)
      ;
    \end{tikzpicture}
  \]
\end{defi}
Translated in the internal language,~$f_1 \approx f_2$ means that if~$x$ and~$y$ are
bisimilar, then~$f_1\,x$ and~$f_2\,y$ are also bisimilar.
It is not difficult to show that
\begin{itemize}
  \item
    $\approx$ is reflexive and symmetric,
  \item
    $\approx$ is compositional: if~$f_1\approx f_2$ and~$g_1 \approx g_2$,
    then~$f_1 g_1 \approx f_2 g_2$ (if the composition makes sense).
\end{itemize}
\label{rk:bisim_equiv}
Without additional properties (which hold in our context), $\approx$ is
however not transitive.

Coinductive types are interpreted by \emph{weakly} terminal coalgebras. There
can be, in principle, several non-isomorphic weakly terminal coalgebras. We
however have the following.
\begin{lem}\label{lem:WTC}
  Let~$T_1$ and~$T_2$ be weakly terminal coalgebras for the endofunctor~$F$,
  \begin{itemize}
    \item
      if~$m:T_1 \to T_2$ is a mediating morphisms, then~$m \approx \id$,

    \item
      if~$f:T_1 \to T_2$ and~$g:T_2\to T_1$ are mediating morphisms,
      then~$gf \approx \id_{T_1}$ and~$fg \approx \id_{T_2}$.
  \end{itemize}
\end{lem}
\begin{proof}
  Consider the following diagram:
  \[
    \begin{tikzpicture}
      \matrix (m) [matrix of math nodes, column sep={6em,between origins}, row
      sep={4em,between origins},text height=1.5ex, text depth=.25ex]
      {
        F(T_2) & F(R) & F(T_1) \\
        T_2    & R    & T_1    \\
        T_2    & R    & T_2    \\
        F(T_2) & F(R) & F(T_2) \\
      };
      \path[morphism]
      (m-1-2) edge node[above]{$F_{r_2}$} (m-1-1)
      (m-1-2) edge node[above]{$F_{r_1}$} (m-1-3)
      (m-2-1) edge node[left]{$t_2$} (m-1-1)
      (m-2-2) edge node[right]{$r$} (m-1-2)
      (m-2-3) edge node[right]{$t_1$} (m-1-3)
      (m-2-2) edge node[above]{$r_2$} (m-2-1)
      (m-2-2) edge node[above]{$r_1$} (m-2-3)
      (m-2-1) edge node[left]{$\id$} (m-3-1)
      (m-2-3) edge node[right]{$m$} (m-3-3)
      (m-3-1) edge node[left]{$t_2$} (m-4-1)
      (m-3-3) edge node[right]{$t_2$} (m-4-3)
      ;
      \path[morphism,densely dashed]
      (m-2-2) edge node[right]{$\id$} (m-3-2)
      (m-3-2) edge node[below]{$r_2$} (m-3-1)
      (m-3-2) edge node[below]{$mr_1$} (m-3-3)
      (m-3-2) edge node[right]{$r$} (m-4-2)
      (m-4-2) edge node[below]{$F_{r_2}$} (m-4-1)
      (m-4-2) edge node[below]{$F_{mr_1}$} (m-4-3)
      ;
    \end{tikzpicture}
  \]
  The only thing needed to make it commutative is that the bottom right square
  is commutative. This follows from the fact that~$m$ is a
  mediating morphism between the weakly terminal coalgebras:
  \[
    \begin{tikzpicture}
      \matrix (m) [matrix of math nodes, column sep={6em,between origins}, row
      sep={4em,between origins},text height=1.5ex, text depth=.25ex]
      {
        R    & T_1    & T_2    \\
        F(R) & F(T_1) & F(T_2) \\
      };
      \path[morphism]
      (m-1-1) edge node[above]{$r_1$} (m-1-2)
      (m-1-2) edge node[above]{$m$} (m-1-3)
      (m-1-1) edge node[left]{$r$} (m-2-1)
      (m-1-2) edge node[left]{$t_1$} (m-2-2)
      (m-1-3) edge node[left]{$t_2$} (m-2-3)
      (m-2-1) edge node[below]{$F(r_1)$} (m-2-2)
      (m-2-2) edge node[below]{$F(m)$} (m-2-3)
      ;
    \end{tikzpicture}
  \]
  The second point is a direct consequence of the first one.
\end{proof}

\begin{asm}\label{asm:bisim}
  We assume our type theory is ``compatible'' with bisimilarity, in the sense
  that any definable function~$f$
  satisfies~$f \approx f$.
\end{asm}
Since it is possible to construct a dependent type theory where bisimilarity
\emph{is} intensional equality~\cite{Ahrens,mogelberg}, this assumption is
reasonable. In our case, it allows to simplify some of the (pen and paper)
proofs while not needing the extension of our type theory.

The only use of this assumption will be Lemma~\ref{lem:bisim_id} on
page~\pageref{lem:bisim_id}, to reduce proofs of~$f\approx g$ to proofs
of~``$f\ x$ and~$g\ x$ are bisimilar for all~$x$''.\footnote{instead of the
much more tedious ``$f\ x_1$ and $g\ x_2$ are bisimilar if~$x_1$ and ~$x_2$
are''}



\section{Indexed Containers}
\label{sec:indexedContainers}

\emph{Indexed containers}~\cite{alti:lics09} were first considered
(implicitly) in type theory 30 years ago by K.~Petersson and
D.~Synek~\cite{treeSet}. Depending on the context and the authors, they are
also called \emph{interaction systems}~\cite{hancock-apal06} or
\emph{polynomial diagrams}~\cite{polyMonads,polyDiagrams}

\begin{defi}
  For~$I:\Set$, an indexed container over~$I$ is a triple~$w = ⟨A,D,n⟩$ where
  \begin{itemize}
    \item
      $i:I \vdash A(i):\Set$,
    \item
      $i:I, a:A \vdash D(i,a):\Set$,
    \item
      $i:I, a:A, d:D \vdash n(i,a,d):I$.\footnote{An abstract, but equivalent,
      way of defining indexed containers over~$I$ is as
      functions~$I\to\Fam\big(\Fam(I)\big)$.}
  \end{itemize}
  In Agda, the definition looks like\footnote{The name \t{IS} comes from
  ``Interaction System'', another name for indexed containers.}
  \begin{allttt}
    record IS (I : Set) where
      field
        A : I → Set
        D : (i : I) → A i → Set
        n : (i : I) → (a : A i) → D i a → I \
  \end{allttt}
\end{defi}
\noindent
A useful intuition is that~$I$ is a set of states and that~$⟨A,D,n⟩$ is a
game:
\begin{itemize}
  \item
    $A(i)$ is the set of \emph{moves} (or \emph{actions} or \emph{commands})
    in state~$i$

  \item
    $D(i,a)$ is the set of \emph{counter-moves} (or \emph{reactions} or
    \emph{responses}) after move~$a$ in state~$i$,

  \item
    $n(i,a,d) : I$ is the \emph{new state} after move~$a$ and counter-move~$d$
    have been played. When no confusion arises, we sometimes write~$i[a/d]$
    for~$n(i,a,d)$.
\end{itemize}
Each indexed container gives rise to a monotone operator on predicates
over~$I$:
\begin{defi}
  If~$w=⟨A,D,n⟩ $ is an indexed container over~$I:\Set$, the \emph{extension
  of~$w$} is the operator:~$\Sem{w} : \Pow(I) \to \Pow(I)$
  where
  \[
    i \IN \Sem{w}(X)
    \quad = \quad
    \SI{a:A(i)}
    \PI{d:D(i,a)}
      i[a/d] \IN X
    \,.
  \]
\end{defi}
\begin{lem}\label{lem:monotonic}
  The operator~$\Sem{w}$ is monotonic, i.e., the following type is inhabited:
  \[
    X\sub Y
    \quad\to\quad
    \Sem{w}(X) \sub \Sem{w}(Y)
  \]
  for every predicates~$X$ and~$Y$ of the appropriate type.
\end{lem}
\begin{proof}\checked
  This is direct: given~$i : X\sub Y$ and~$\langle a,f\rangle : \Sem{w}(X)$,
  we just need to ``carry''~$i$ through~$\Sem{w}$ and return~$\langle a , i
  \circ f\rangle$.
  \end{proof}
Indexed containers form the objects of several categories of
interest~\cite{hancock-apal06,polyMonads,alti:lics09,polyDiagrams}, but that
will only play a very minor role here.

Many indexed containers of interest have singleton actions or singleton
reactions.
\begin{defi}
  An indexed container~$\langle A, D, n\rangle$ is called
  \begin{itemize}
    \item
      \emph{angelic} if $D(i,a)$ is always (isomorphic to) a singleton type,

    \item
      \emph{demonic} if $A(i)$ is always (isomorphic to) a singleton type,

    \item
      \emph{lopsided} if it is either angelic or demonic.

  \end{itemize}
\end{defi}


\subsection{Composition}  
\label{subsub:indexedComposition}

Extensions of indexed containers can be composed as functions. There is a
corresponding operation on the indexed containers.
\begin{defi}
  If~$w_1 = ⟨ A_1,D_1,n_1⟩ $ and~$w_2=⟨ A_2,D_2,n_2⟩ $ are two indexed containers
  on~$I$, the \emph{composition} of~$w_1$ and~$w_2$ is the
  indexed container~$w_2\circ w_1 = ⟨ A , D , n⟩ $ where
  \begin{itemize}
    \item $A(i) = \SI{a_1:A_1(i)}\PI{d_1:D_1(i,a_1)}
      A_2\big(i_1[a_1/d_1]\big) = i \IN \Sem{w_1}(A_2)$,
    \item $D\big(i,⟨ a_1,f⟩ \big) = \SI{d_1:D_1(i,a_1)}D_2\big(i[a_1/d_1],
    f\,d_1\big)$,
      \item $n\big(i,⟨ a_1,f⟩ ,⟨ d_1,d_2⟩ \big) = n_2\big(n_1(i,a_1,d_1),f\,d_1,d_2\big)$.
  \end{itemize}
\end{defi}
\begin{lem}\label{lem:composition}
  For every indexed containers~$w_1$ and~$w_2$ and predicate~$X$, we have
  \[
    \Sem{w_2} \circ \Sem{w_1}(X)
    ==
    \Sem{w_2\circ w_1}(X)
  \]
  where~$Y == Z$ is an abbreviation for~$(X\sub Y)\times(Y\sub X)$.

  If function extensionality holds, the pair of functions are inverse to each
  other.
\end{lem}
\begin{proof}\checked
The main point is that the intensional axiom of choice
\[
  \PI{d: D}\SI{a: A(d)} \varphi(d,a)
  \quad \leftrightarrow \quad
  \SI{f: \PI{d:D} A(d)}\PI{d: D} \varphi(d,f\,d)
\]
is provable in Martin-L\"of type theory.
We then have
\begin{myequationqed}
  i\IN \Sem{w_2} \circ \Sem{w_1}(X)
  & = &
     \SI{a_1:A_1(i)} \PI{d_1:D_1(i,a_1)} \SI{a_2:A_2(i[a_1/d_1])}\\
  && \PI{d_2:D_2(i[a_1/d_1],a_2)}
     i[a_1/d_1][a_2/d_2] \IN X\\
  \text{\footnotesize(axiom of choice)}
  &\leftrightarrow&
     \SI{a_1} \SI{f:\PI{d_1:D_1(i,a_1)}A_2(i[a_1/d_1])} \PI{d_1}\\
  && \PI{d_2} i[a_1/d_1][f\,d_1/d_2] \IN X\\
  &\leftrightarrow&
     \SI{⟨a_1,f⟩} \PI{⟨d_1,d_2⟩} i[a_1/d_1][f\,d_1/d_2] \IN X\\
  & = &
  i\IN \Sem{w_2 \circ w_1}(X)
  \,.
\end{myequationqed}
\end{proof}


\subsection{Duality}  

\begin{defi}\label{def:duality}
If~$w=⟨ A,D,n⟩ $ is an indexed container over~$I$, we write~$w^\perp$ for the
indexed container~$⟨ A^\perp,D^\perp,n^\perp⟩ $ where
\begin{itemize}
  \item 
    $A^\perp(i) = \PI{a:A(i)}D(i,a)$,

  \item
    $D^\perp (i, \BLANK) = A(i)$,
    (note that it does not depend on the value of~$f : A^\bot(i)$)

  \item
    $i[f/a] = i[a/f\,a]$.
\end{itemize}
\end{defi}
\begin{lem}\label{lem:duality}
  For every indexed container~$w=⟨ A,D,n⟩ $, the following type is inhabited:
  \[
    i \IN \Sem{w^\bot}(X)
    \quad\longleftrightarrow\quad
    \PI{a:A(i)}\SI{d:D(i,a)} i[a/d] \IN X
    \,.
  \]
  With function extensionality, this is an isomorphism.
\end{lem}
\begin{proof}[Sketch of proof]\checked
Just like Lemma~\ref{lem:composition}, the proof relies on the intensional
axiom of choice, which shows that
\[
  \SI{f:A^\bot(i)}
  \PI{a:D^\bot(i,f)}
  \varphi(a, f\,a)
  \leftrightarrow
  \PI{a:A(i)}
  \SI{d:D(i,a)}
  \varphi(a, d)
  \,.\qedhere
\]
\end{proof}
It is interesting to note that lopsided containers are closed under duality,
and that duality is involutive on them.\label{rk:lopsided} For any family~$X$ indexed
by~$I$ and~$t:\PI{i:I}X(i)\to I$, the following containers are dual to each
other:
\begin{itemize}
  \item $A(i) = X(i)$,
  \item $D(i,x) = \{\star\}$,
  \item $i[x/\star] = t(i,x)$.
\end{itemize}
and
\begin{itemize}
  \item $A(i) = \{\star\}$,
  \item $D(i,\star) = X(i)$,
  \item $i[\star/x] = t(i,x)$.
\end{itemize}


\subsection{Free Monad}  
\label{subsub:freeMonad}

If~$w=⟨ A,D,n⟩ $ is an indexed container on~$I$, we can consider the free
monad generated by~$\Sem{w}$. N. Gambino and M. Hyland proved that the free
monad~$F_w$ generated by some~$\Sem{w}$ (a dependent polynomial
functor) is of the form~$\Sem{w^\ast}$ for some indexed
container~$w^\ast$~\cite{GambHyl}. It is characterized by the fact
that~$\Sem{w^\ast}(X)$ is the least
fixed point of~$Y\mapsto X \cup \Sem{w}(Y)$. In other words, we are looking
for an indexed container~$w^\ast$ satisfying
\begin{itemize}
  \item
    $X \cup \Sem{w}\big(\Sem{w^\ast}(X)\big) \sub \Sem{w^\ast}(X)$
  \item
    $X \cup \Sem{w}(Y) \sub Y \to \Sem{w^\ast}(X) \sub Y$.

\end{itemize}
Informally,~$i\IN\Sem{w}^\ast(X)$ iff
$i \IN X \cup w\big(X \cup w( X \cup w(X \cup \dots))\big)$.
Expending the definition of~$\Sem{w}$, this means that
\[
  \SI{a_1}\PI{d_1}\SI{a_2}\PI{d_2}
  \ \dots\ \ i[a_1/d_1][a_2/d_2]\dots \IN X
  \,.
\]
This pseudo formula depicts a well-founded tree (with branching
on~$d_i:D(\BLANK,\BLANK)$) in which branches are finite: either they end at
a~$\Pi{d_i}$ because the corresponding domain~$D(\BLANK,\BLANK)$ is empty, or
they end in a state~$i[a_1/d_1][a_2/d_2]\dots$ which belongs to~$X$. The
length of the sequence~$a_1/d_1,a_2/d_2,\dots$ is finite but may depend on
what moves / counter-moves are chosen as the sequence grows longer.
We can define~$w^\ast$ inductively.
\begin{defi}
  Define~$w^\ast = ⟨A^\ast,D^\ast,n^\ast⟩$ over~$I$ as:
  \begin{itemize}
    \item
      $A^\ast : \Pow(I)$ is a weak initial algebra for the
      endofunctor~$X \mapsto \LAM{i} \big(\One + \Sem{w}(X)(i)\big)$
      on~$\Pow(I)$.
      Concretely, there are two constructors for~$A^\ast(i)$:

      \[
        \Rule{}{\leaf : A^\ast(i)}
        \quad
        \text{and}
        \quad
        \Rule{ a : A(i)
               \qquad
               k : \PI{d : D(i,a) } A^\ast(n\,i\,a\,d) }
             { \node(a,k) : A^\ast(i) }
        \,.
      \]
    Thus, an element of~$A^\ast(i)$ is a well-founded tree where each internal
    node is labeled with an elements~$a:A(i)$ and the branching is given by
    the set~$D(i,a)$.

    \item The components~$D^\ast$ and~$n^\ast$ are defined by recursion:
    \begin{itemize}
      \item in the case of a~$\leaf$:
        \[\begin{array}{lclclcl}
          D^\ast \big(i,\leaf(i)\big) & = & \One : \Set \\
          n^\ast\big(i,\leaf(i),\star\big)  & = & i : I \\
        \end{array}\]

      \item in the case of a~$\node$:
        \[\begin{array}{lcl}
          D^\ast \big(i,\node(a,k)\big)
            & = &
          \SI{d : D (i,a)} D^\ast\big(i[a/d],f\,d\big) : \Set\\
          n^\ast\big(i,\node(a,k),⟨ d,d'⟩ \big)
            & = &
          n^\ast\big(i[a/d],k\,d,d'\big) : I
          \,.
    \end{array}
    \]
    \end{itemize}
    The corresponding Agda definition is
    \begin{allttt}
      module FreeMonad (I : Set) (w : IS I) where
        open IS w

        data A* : I → Set where
          Leaf : (i : I) → A* i
          Node : (i : I) → (a : A i) → (f : (d : (D i a)) → A* (n a d)) → A* i

        D* : (i : I) → (t : A* i) → Set
        D* i Leaf = One
        D* i (Node a f) = \(\Sigma\) (D i a) (λ d → D* (n a d) (f d))

        n* : (i : I) → (t : A* i) → (b : D* i t) → I
        n* i Leaf ⋆ = i
        n* i (Node a f) ( d , ds ) = n* (f d) ds
    \end{allttt}

  \end{itemize}
\end{defi}
In the presence of extensional equality, this particular inductive definition
can be encoded using standard~$W$-types~\cite{polyMonads}. As it is given, it
avoids using equality but needs either a universe, or degenerate\footnote{in
the sense that the inductive set~$A^\ast(i)$ does not depend on the recursive
functions~$D^\ast(i,\BLANK)$ and~$n^\ast(i,\BLANK,\BLANK)$.} induction-recursion
on~$\Fam(I)$.
This construction does indeed correspond to the free monad described by N.
Gambino and M. Hyland:
\begin{lem}\label{lem:RTC_mu}
  If~$w$ is a container indexed on~$I:\Set$, we have
\begin{enumerate}
  \item
    for all~$X:\Pow(I)$, $X \cup \Sem{w}\big(\Sem{w^\ast}(X)\big) \sub \Sem{w^\ast}(X)$
  \item
    for all~$X,Y:\Pow(I)$, $X \cup \Sem{w}(Y) \sub Y \to \Sem{w^\ast}(X) \sub Y$.
\end{enumerate}
\end{lem}
\begin{proof}\checked
\end{proof}

\subsection{Greatest Fixed Points}  

Agda has some support for infinite values via the~$\infty\BLANK$ type
constructor making a type ``lazy'', \ie stopping computation. Using this and
the operators~$\natural: A \to \infty A$ (to freeze a value) and~$\flat :
\infty A \to A$ (to unfreeze it), it is possible to define the
type~$\nu_{\Sem{w}}$ (which we'll write~$\nu_w$) for any indexed
container~$w$. The termination checker used in Agda~\cite{foetus} also checks
productivity of recursive definition, but since it is not clear that this is
sound when inductive and coinductive types are
mixed~\cite{ThorstenNad,PHTotality} we will only use the standard introduction
and elimination rules in our developments:
\[
  \Rule{\sigma : X \sub \Sem{w}(X)}{\nuIntro\,\sigma : X \sub \NU{w}}
  \qquad\text{and}\qquad
  \Rule{}{\nuElim : \NU{w} \sub \Sem{w}(\NU{w})}
\]
which are definable in Agda.

Elements of~$\NU{w}$ are formed by coalgebras for~$\Sem{w}$,
and any element of~$i\IN\NU{w}$ can be ``unfolded''
into an element of~$i \IN \Sem{w}(\NU{w})$ \ie into an element of
\[
  \SI{a : A(i)}
  \PI{d : D(i,a)}
    i[a/d] \IN \NU{w}
  \,.
\]
We can repeat this unfolding and informally decompose an element of~$i\IN\NU{w}$
into an infinite ``lazy'' process of the form
\[
  \SI{a_1 : A(i)}
  \PI{d_1 : D(i,a_1)}
  \SI{a_2 : A(i[a_1/d_1])}
  \PI{d_2 : D(i[a_1/d_1],a_2)}
  \cdots
\]
We therefore picture an element of~$i\IN\NU{w}$ as an infinite tree (which
need not be well-founded). Each node of such a tree has an implicit state
in~$I$, and the root has state~$i$. If the state of a node is~$j$, then the
node contains an element~$a$ of~$A(j)$, and the branching of that node is
given by~$D(j,a)$. Note that some finite branches may be inextensible when
they end at a node of state~$j$ with label~$a$ for which~$D(j,a)$ is empty.

\subsubsection*{Examples} 

We will be particularly interested in fixed points of the form~$\NU{w^\bot}$.
Because of Lemma~\ref{lem:duality}, an element of~$i\IN\NU{w^\bot}$ unfolds to
a potentially infinite object of the form
\[
  \PI{a_1 : A(i)}
  \SI{d_1 : D(i,a_1)}
  \PI{a_2 : A(i[a_1/d_1])}
  \SI{d_2 : D(i[a_1/d_1],a_2)}
  \cdots
\]
For such types, the branching comes from~$A(\BLANK)$ and the labels
from~$D(\BLANK,\BLANK)$. Here are some example of the kind of objects we get.

\begin{enumerate}
  \item
    Streams on~$X$ are isomorphic to~``$\star \IN \NU{w^\bot}$''
    where~$I=\One = \{\star\}$ and~$w=⟨A,D,n⟩$ with
    \begin{itemize}
      \item
        $A(\star) = \One$,

      \item
        $D(\star,\star) = X$,

      \item
        $n(\star,\star,x) = \star$.
    \end{itemize}

  \item
    Increasing streams of natural numbers are isomorphic to~``$0\IN\NU{w^\bot}$''
    where~$I=\N$ and~$w=⟨A,D,n⟩$ with
    \begin{itemize}
      \item
        $A(i) = \One$,

      \item
        $D(i,\star) = \SI{j:\N} (i < j)$,

      \item
        $n(i,\star,⟨j,p⟩) = j$.

    \end{itemize}

  \item
    Infinite, finitely branching trees labeled by~$X$ are isomorphic to~$\star \IN \NU{w^\bot}$
    where~$I=\One$ and~$w=⟨A,D,n⟩$ to be
    \begin{itemize}
      \item
        $A(\star) = \SI{k:\N} N(k)$, where~$N(k)$ is the set with exactly~$k$ elements,

      \item
        $D(\star,\langle k, i\rangle) = X$,

      \item
        $n(\star,k,x) = \star$.
    \end{itemize}

  \item
    In general, non-losing strategies for the second player from state~$i:I$
    in game~$w$ are given by~$i \IN \NU{w^\bot}$.
\end{enumerate}


\subsubsection*{Bisimulations}  

The appropriate equivalence relation on coinductive types is
\emph{bisimilarity}. Translating the categorical notion from
page~\pageref{def:categorical_bisimulation} for the
type~$\NU{w}$,\footnote{Note that we only define bisimilarity for elements
of~$i\IN\NU{w}$, and not for arbitrary ~$\Sem{w}$-coalgebras.} we get
that~$T_1 : i_0 \IN \NU{w}$ is bisimilar to~$T_2 : i_0 \IN \NU{w}$ if:
\begin{enumerate}
  \item
    there is an~$I$ indexed family~$R_i : (i\IN\NU{w}) \times (i\IN\NU{w}) \to
    \Set$ (an ``indexed relation'') s.t.

  \item
    $R_{i_0} \langle T_1,T_2\rangle$ is inhabited,

  \item
    whenever~$R_i \langle T_1, T_2\rangle$ is inhabited, we have
    \begin{itemize}
      \item
        $a_1 \equiv a_2$,

      \item
        for every $d_1:D(i,a_1)$, the elements~$f_1\,i[a_1/d_1]$
        and~$f_2\,i[a_2/d_2]$ are related by~$R$,
    \end{itemize}
    where~$\langle a_1 , f_1\rangle$ [resp. $\langle a_2 , f_2\rangle$] comes
    from the coalgebra structure~$\NU{w} \sub \Sem{w}\NU{w}$ applied to~$T_1$
    [resp.~$T_2$].
\end{enumerate}
Expressing this formally is quite verbose as values need to be
transported along equalities to have an appropriate type. For example,
having~$a_1 \equiv a_2 \in A(i)$ only entails that~$D(i,a_1)$ and~$D(i,a_2)$
are isomorphic, and thus, $d_1 : D(i,a_1)$ is not strictly speaking an element
of~$D(i,a_2)$!

We noted on page~\pageref{rk:bisim_equiv} that categorically
speaking, without hypotheses on the functor~$F$,~$\approx$ was not necessarily
transitive. We have
\begin{lem}
  If~$w$ is a container indexed on~$I$, then~$\approx$ is an equivalence
  relation on any~$i\IN\NU{w}$:
  \begin{itemize}
    \item
      for any~$T:i\IN\NU{w}$, there is an element in~$T \approx T$,

    \item
      for any~$T_1,T_2:i\IN\NU{w}$, there is a function in~$(T_1 \approx T_2)
      \to (T_2 \approx T_1)$,

    \item
      for any~$T_1,T_2,T_3:i\IN\NU{w}$, there is a function in~$(T_1 \approx
      T_2) \to (T_2 \approx T_3) \to (T_1 \approx T_3)$.

  \end{itemize}
\end{lem}
\begin{proof}\checked
  The result is intuitively obvious but while reflexivity is easy, proving
  transitivity (and to a lesser extent symmetry) in Agda is surprisingly
  tedious. Explaining the formal proof is probably pointless as it mostly
  consists of transporting elements along equalities back and forth.
\end{proof}
We will keep some of the bisimilarity proofs in the meta theory in order to
simplify the arguments. The only consequence of the assumption
from page~\pageref{asm:bisim} that we'll need is the following.
\begin{lem}\label{lem:bisim_id}
  Suppose that~$f,g : i_1 \IN \NU{w_1} \to i_2 \IN \NU{w_2}$ are definable in type
  theory,
  then, to prove that~$f \approx g$,\footnote{\ie $S\approx T \to f\,S \approx
  g\,T$} it is enough to show that~$f\,T \approx g\,T$ for any~$T:i_1 \IN
  \NU{w_1}$.
\end{lem}
\begin{proof}
  If~$S \approx T$, we have~$f\,S \approx f\,T \approx g\,T$ where the first
  bisimulation comes from the assumption~$f\approx f$ from
  page~\pageref{asm:bisim}
  and the second bisimulation comes from the hypothesis of the lemma.
\end{proof}
This makes proving~$f\approx g$ simpler as we can replace the hypothesis~$T_1
\approx T_2$ by the stronger~$T_1 \equiv T_2$.

\subsubsection*{Weakly Terminal Coalgebras} 

We will have to show that some sets are isomorphic ``up to bisimilarity''. To
do that, we'll use Lemma~\ref{lem:WTC} by showing that the two sets are weakly
terminal coalgebras for the same functor~$\Sem{w}$. (One of the sets will
always be~$\NU{w}$, making half of this automatic.)

To show that~$T : \Pow(I)$ is a weakly terminal coalgebra for~$\Sem{w}$, we
have to define, mimicking the typing rules for coinductive types:
\begin{itemize}
  \item
    $\mathtt{elim} : T \sub \Sem{w}(T)$,

  \item
    $\mathtt{intro} : X \sub \Sem{w}(X) \to X \sub T$,

  \item
    $\mathtt{comp}_{X,c,x,i} : \mathtt{elim}(\mathtt{intro}\, c\, i\, x)
                     \equiv
                     \Sem{w}_{\mathtt{intro}\,c}i\,(c\,i\,x)$
    whenever
    \begin{itemize}
      \item
        $X:\Pow(I)$,
        $c:X\sub\Sem{w}(X)$,
        $i:I$ and~$x:i\IN X$,
      \item
        $\Sem{w}_{\mathtt{intro}\,c} : \Sem{w}(X) \sub \Sem{w}(T)$ comes from
        Lemma~\ref{lem:monotonic}.
    \end{itemize}
\end{itemize}
By Lemma~\ref{lem:WTC}, we have
\begin{cor}\label{cor:WTC}
  If $C$ is a weakly terminal coalgebra for~$\Sem{w}$, then there are
  functions
  $f : \NU{w} \sub C$ and~$g : C \sub \NU{w}$ such that
  \[
  f\,i\,(g\,i\,T) \approx T
  \]
  for any~$T : i \IN \NU{w}$.
\end{cor}
\begin{proof}\checked
  This is the second point of Lemma~\ref{lem:WTC}, and it has been
  formalized in Agda.
\end{proof}



\section{Simulations and Evaluation}

\subsection{Functions on Streams} 

Continuous function from~$\Stream{A}$ to~$\Stream{B}$ can be described by
infinite, $A$-branching ``decision trees'' with two kinds of nodes: $\INPUT$
and~$\OUT_b$ with~$b\in B$. The idea is that~$f(s) = [b_1, b_2, \dots]$ if and
only if the infinite branch corresponding to~$s$ contains, in order, the
nodes~$\OUT_{b_1}$, $\OUT_{b_2}$, \dots, interspersed with~$\INPUT$ nodes. For
that to be well defined, we need to guarantee that all infinite branches in
the tree contain infinitely many~$\OUT$ node, or equivalently, that no branch
contains infinitely many \emph{consecutive}~$\INPUT$ nodes.

\begin{thm}[\cite{hancock09:_repres_of_stream_proces_using}]\label{thm:stream_transducer}
  The set of continuous functions from~$\Stream{A}$ to~$\Stream{B}$ is
  isomorphic to the set~$\nu X.\mu Y. (B \times X) + (A \to Y)$.
\end{thm}
The nested fixed points guarantee that along a branch, there can only be
finitely many consecutive \tINPUT\ nodes:
\begin{itemize}
  \item
    $\mu Y. \, (B \times X) + (A \to Y)$: well-founded~$A$-branching trees with
    leafs in~$B \times X$, \ie consisting of an $\OUT_b$ node and an element
    of~$X$,

  \item
    $\nu X. \cdots$ the element in~$X$ at each leaf for the well-founded trees
    is another such well-founded tree, ad infinitum.
\end{itemize}
We can evaluate each such tree into a continuous function, a process
colloquially referred to as ``stream eating'': we consume the elements of the
stream to follow a branch in the infinite tree, and output~$b$ (an element of the
result) whenever we find an~$\OUT_b$ node.\footnote{The converse mechanism, of
converting a function into  process into a tree cannot be defined in type
theory, see the discussion on page~\pageref{sub:completeness}.}

\medbreak
Our aim is to extend Theorem~\ref{thm:stream_transducer} to coinductive types
of the form~$i \IN \NU{w^\bot}$. The problem is doing so in a
type-theoretic manner and even the case of dependent streams (where the type
of an element may depend on the values of the previous elements) is not trivial.
Retrospectively, the difficulty was that there are two generalizations of
streams:
\begin{itemize}
  \item
    adding states: we consider dependent streams instead of streams;
  \item
    adding branching: we consider trees instead of streams.
\end{itemize}
Both cases required the introduction of states \emph{and} branching, making
those seemingly simpler cases as hard as the general one.


\subsection{Linear Simulations as Transducers} 

We are going to define a notion of ``transducer'' that can explore some branch
of its input and produce some output along the way. Such notions have already
been considered as natural notions of morphisms between dependent containers:
linear simulations~\cite{polyDiagrams} and general
simulations~\cite{hancock-apal06}. These notions generalize morphisms as
representations of pointwise inclusions~$\Sem{w_1} \sub \Sem{w_2}$ (called
cartesian strong natural transformations), which only make sense for
containers with the same index set~\cite{polyMonads,GambHyl,alti:lics09}.

A transducer from type~$i_1 \IN \NU{w_1^\bot}$ to type~$i_2 \IN \NU{w_2^\bot}$
works as follows: given an argument (input) in~$i_1 \IN \NU{w_1^\bot}$,
morally of the form
\[
  \PI{a_0}\SI{d_0}\PI{a_1}\SI{d_1} \dots
\]
it must produce (output) an object of the form
\[
  \PI{b_0}\SI{e_0}\PI{b_1}\SI{e_1} \dots
\]
In other words, the transducer
\begin{itemize}
  \item
    consumes~$b$ (they are given by the environment when constructing the
    result) and~$d$ (they may be produced internally by the input),

  \item
    produces~$e$ (as part of the output) and~$a$ (to be used internally by
    feeding them to the input).
\end{itemize}
Graphically, the transducer ``plugs into'' the input, producing the output:
  
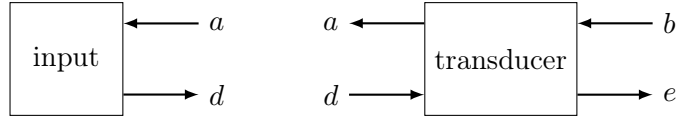
\[
   \begin{tikzpicture}[
       box/.style={draw, minimum size=1.5cm,
       }]

       \node[box] (a) {\text{transducer}};
       \node[box,left=4cm of a] (input) {input};

       \draw[latex-,thick] (input.32) --++(0:1cm) node [right] {$a$};
       \draw[-latex,thick] (input.-32) --++(0:1cm) node [right] {$d$};

       \draw[-latex,thick] (a.-205) --++(0:-1cm) node [left] {$a$};
       \draw[latex-,thick] (a.205) --++(0:-1cm) node [left] {$d$};

       \draw[latex-,thick] (a.25) --++(0:1cm) node [right] {$b$};
       \draw[-latex,thick] (a.-25) --++(0:1cm) node [right] {$e$};


   \end{tikzpicture}
  \]

A very simple kind of transducer works as follows:
\begin{enumerate}
  \item
    when given some~$b_0$,
  \item
    it produces an~$a_0$ and feeds it to its argument,
  \item
    it receives a~$d_0$ from its argument,
  \item
    and produces an~$e_0$ for its result.
  \item
    It starts again at step (1) possibly in a different internal state.
\end{enumerate}
The resulting interface is of the form
  
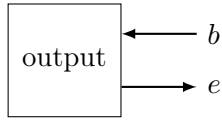
\[
   \begin{tikzpicture}[
       box/.style={draw, minimum size=1.5cm,
       }]

       \node[box] (output) {output};

       \draw[latex-,thick] (output.25) --++(0:1cm) node [right] {$b$};
       \draw[-latex,thick] (output.-25) --++(0:1cm) node [right] {$e$};

   \end{tikzpicture}
  \]

This intuition is captured by the following definition.
\begin{defi}
  Let~$w_1=⟨A_1,D_1,n_1⟩$ and~$w_2=⟨A_2,D_2,n_2⟩$ be indexed containers over~$I_1$ and~$I_2$
  and let~$R:\Pow(I_1 \times I_2)$ be a relation between states. We say that~$R$ is
  a \emph{linear simulation} from~$w_1$ to~$w_2$ if it comes with a proof:
  \[\begin{array}{rcl}
    \rho \quad:\quad \PI{i_1:I_1} \PI{i_2:I_2} R(i_1,i_2) &\to&
          \PI{a_2 : A_2(i_2)} \\
      &&  \SI{a_1 : A_1(i_1)} \\
      &&  \PI{d_1 : D_1(i_1,a_1)} \\
      &&  \SI{d_2 : D_2(i_2,a_2)} \\
      && \quad R\big(i_1[a_1/d_1], i_1[a_2/d_2]\big)
      \,.
  \end{array}\]
  We write~$(R,\rho) : w_1 \linear w_2$, but usually leave the~$\rho$
  implicit and write~$R:w_1\linear w_2$.
\end{defi}
We will not need this fact in the present paper, but this notion of simulation
enjoys a strong universal property. It arises as the adjoint of a very natural
tensor product~\cite{polyDiagrams} (see also the discussion on
page~\pageref{sub:internal_sim}).

To justify the fact that this can serve as a transducer, we need to
``evaluate'' a simulation on elements of~$i\IN\NU{w_1^\bot}$.
\begin{lem}\label{lem:eval_linear}
  Let~$w_1=⟨A_1,D_1,n_1⟩$ and~$w_2=⟨A_2,D_2,n_2⟩$ be indexed containers over~$I_1$ and~$I_2$,
  and~$(R,\rho) : w_1 \linear w_2$. We have a function
  \[
  \eval_R : \PI{i_1:I_1}\PI{i_2:I_2}R(i_1,i_2) \to
             (i_1 \IN \NU{w_1^\bot}) \to
             (i_2 \IN \NU{w_2^\bot})
  \,.
  \]
\end{lem}
\begin{proof}\checked
  This amounts to unfolding the simulation as a linear transducer. The main
  point in the Agda proof is to show
  that the predicate~$\NU{w_1^\bot}\between R^\sim(i_2)$ is a
  coalgebra for~$\Sem{w_2^\bot}$.
\end{proof}

The next lemma shows that this notion of simulation gives an appropriate
notion of morphism between indexed containers.
\begin{lem}
  We have:
  \begin{itemize}
    \item the identity type on~$I$ is a linear simulation from any~$w$
      over~$I$ to itself,
    \item if~$R$ is a linear simulation from~$w_1$ to~$w_2$, and if~$S$ is a
      linear simulation from~$w_2$ to~$w_3$, then~$S\circ R$ is a linear
      simulation from~$w_1$ to~$w_3$, where~$S\circ R$ is the relational
      composition of~$R$ and~$S$:
      \[
        (S\circ R)(i_1,i_3)
        \quad=\quad
        \SI{i_2 : I_2} R(i_1,i_2) \times S(i_2,i_3)
      \]
  \end{itemize}
\end{lem}
\begin{proof}[Proof] \checked
  That the identity type is a linear simulation is straightforward. Composing
  simulation amounts to extracting the functions~$a_1 \mapsto a_2$ and~$d_2
  \mapsto d_1$ from the simulations, and composing them.
\end{proof}
Note that composition of simulations is only associative \emph{up to
associativity of relational composition} (pullback of spans) so that a quotient is needed to turn
indexed containers into a category~\cite{polyDiagrams}.
What is nice is that composition of simulations corresponds to composition of
their evaluations, up to bisimilarity.
\begin{lem}
  If~$w_1$, $w_2$ and~$w_3$ are containers indexed on~$I_1$, $I_2$ and~$I_3$,
  and if~$R:w_1 \linear w_2$ and~$S: w_2 \linear w_3$, then we have
  \[
    \eval_{S\circ R}\,i_1\,i_3\,⟨i_2,⟨r,s⟩⟩
    \quad\approx\quad
    \eval_S\,i_2\,i_3\,s \circ \eval_R\,i_1\,i_2\,r
  \]
  where
  \begin{itemize}
    \item
      $i_1:I_1$, $i_2:I_2$, $i_3:I_3$,

    \item
      $r:R(i_1,i_2)$ and~$s:S(i_2,i_3)$,

    \item
      and thus, $⟨i_2,⟨r,s⟩⟩: (S\circ R)(i_1,i_3)$.
  \end{itemize}
  Recall that for functions, $f \approx g$ means that for every input~$T$
  (here of type~$i_1\IN\NU{w_1^\bot}$), we have ``$f\,T \approx g\,T$, \ie
  $f\,T$ is bisimilar to~$g\,T$''.
\end{lem}
\begin{proof}\checked
  This is one instance where the direct, type theoretic proof of bisimilarity
  is possible, and not (too) tedious. With the transducer intuition in mind, this result
  is natural: starting from some~$a_3 : A_1(i_3)$, we can either
  \begin{itemize}
    \item
      transform it to~$a_2 : A_2(i_2)$ (with the simulation from~$w_2$
      to~$w_3$) and then to~$a_1 : A_1(i_1)$ (with the simulation from~$w_1$
      to~$w_2$),
    \item
      or transform it directly to~$a_1:A(i_1)$ (with the composition of the
      two simulations).
  \end{itemize}
  Because composition is precisely defined by composing the functions
  making the simulations, the two transformations are obviously equal. (The
  Agda proof is messier than that but amounts to the same thing.)

  Note that because of Lemma~\ref{lem:bisim_id}, the Agda proof only needs to
  show that~$(\eval_S \circ \eval_R) \, T$ is bisimilar to~$\eval_{S\circ R}
  \, T$.
\end{proof}

\subsection{General Simulations}  

As far as representing functions from~$i_1\IN\NU{w_1^\bot}$
to~$i_2\IN\NU{w_2^\bot}$, linear simulation are not very powerful. For streams,
the first~$n$ elements of the result may depend at most on the first~$n$
elements of the input! Here is a typical continuous function that cannot be
represented by a linear simulation. Given a stream~$s$ of natural numbers,
look at the head of~$s$:

\begin{itemize}
  \item
    if it is~$0$, output~$0$ and start again with the
    tail of the stream,

  \item
    if it is~$n>0$, remove the next~$n$ element of the stream, output their
    sum, and start again.
\end{itemize}
For example, on the stream~$[0,1,2,3,4,5,6,\dots]$, the function outputs
\[
  \Big[0,\,
  2,\,
  \overbrace{\underbrace{\big.4+5+6}_{=15}}^{\text{3 elements}},\,
  \overbrace{\underbrace{\big.8+9+\cdots+14}_{=77}}^{\text{7 elements}},\,
  \overbrace{\underbrace{\big.16+17+\cdots+31}_{=376}}^{\text{15 elements}},\,
  \dots\Big]
  = [0,2,15,77,376, \dots]
\]
We can generalize transducers by allowing them to work in the following
manner.
\begin{enumerate}
  \item
    When given some~$b_0$,
  \item
    they produce an~$a_0$ and feed it to their argument,
  \item
    they receive a~$d_0$ from their argument,
  \item
    and \emph{either go back to step~\textup{(2)} or} produce an~$e_0$
    (output) for their result,
  \item
    start again at step (1)...
\end{enumerate}
In other words, steps (2) and (3) can occur several times in a row. We can
even allow the transducer to go directly from step~(1) to step~(4), bypassing
steps~(2) and~(3) entirely. Of course steps~(3) and~(4) should never happen
infinitely many times consecutively.
\begin{defi}
  Let~$w_1$ and~$w_2$ be indexed containers
  over~$I_1$ and~$I_2$, let~$R:\Pow(I_1 \times I_2)$ be a relation between
  states; we say that~$R$ is a \emph{general simulation from~$w_1$ to~$w_2$}
  if it is a linear simulation from~$w_1^\ast$ to~$w_2$.
\end{defi}
In other words, $\langle R,\rho\rangle$ is a general simulation
from~$⟨A_1,D_1,n_1⟩$ to~$⟨A_2,D_2,n_2⟩$ if
\[\begin{array}{rcl}
  \rho : \PI{i_1:I_1}\PI{i_2:I_2}R(i_1,i_2) &\to&
        \PI{a_2 : A_2(i_2)} \\
    &&  \SI{\alpha_1 : A_1^\ast(i_1)} \\
    &&  \PI{\delta_1 : D_1^\ast(i_1,\alpha_1)} \\
    &&  \SI{d_2 : D_2(i_2,a_2)} \\
    && \quad R\big(i_1[\alpha_1/\delta_1], i_2[a_2/d_2]\big)
    \,.
\end{array}\]

Thanks to Lemma~\ref{lem:eval_linear}, such a simulation automatically gives
rise to a function from~$\NU{w_1^{\ast\bot}}$ to~$\NU{w_2^\bot}$. Fortunately,
$\NU{w_1^{\ast\bot}}$ is, up to bisimulation, isomorphic to~$\NU{w_1^\bot}$.
\begin{lem}\label{lem:*WTC}
  $\NU{w^{\ast\bot}}$ is a weakly terminal coalgebra for~$\Sem{w^\bot}$.
  
\end{lem}
\begin{proof}\checked
  From an element of~$i \IN \NU{w^{\ast\bot}}$  we can use the elimination rule
  and extract a member of~$\PI{\alpha:A^\ast(i)}\SI{\delta:D^\ast(i,\alpha)}
  i[\alpha/\delta]\NU{w^{\ast\bot}} $.
  Given some~$a:A(i)$, we instantiate~$\alpha$ to~$\node(a,\LAM{d:D(i,a)}
  \leaf) : A^\ast(i)$ (a single~$a$, followed by nothing), and its responses
  are just responses to~$a$. This produces an element
  of~$\PI{a:A(i)}\SI{d:D(i,a)} i[a/d] \IN \NU{w^{\ast\bot}}$, \ie an element
  of~$i\IN\Sem{w^\bot}(\NU{w^{\ast\bot}})$. We've just shown
  that~$\NU{w^{\ast\bot}} \sub \Sem{w^\bot}\NU{w^{\ast\bot}}$. We refer to the
  Agda code for the rest of the proof.
\end{proof}
\begin{cor}\label{cor:RTCbot_RTC_WTC}
  $\NU{w^{\ast\bot}}$ and~$\NU{w^\bot}$ are isomorphic up to bisimulation:
  \[
  \NU{w^{\ast\bot}} \stackrel{\approx}\longleftrightarrow \NU{w^\bot}
  \,.
  \]
\end{cor}
\begin{proof}
  This is a direct consequence of Lemma~\ref{lem:WTC}.
\end{proof}
By composing the function~$\NU{w^\bot} \sub \NU{w^{\ast\bot}}$ with the
evaluation of linear simulations (Lemma~\ref{lem:eval_linear}), we get an
evaluation function for general simulations.
\begin{cor}\label{cor:eval*}
  Let~$w_1=⟨A_1,D_1,n_1⟩ $ and~$w_2=⟨A_2,D_2,n_2⟩ $ be indexed containers over~$I$
  and~$I_2$, and~$(R,\rho) : w^\ast \linear w_2$. We have a function
  \[
  \eval^\ast_R : \PI{i_1:I_1}\PI{i_2:I_2}R(i_1,i_2) \to
             (i_1 \IN \NU{w_1^\bot}) \to
             (i_2 \IN \NU{w_2^\bot})
  \,.
  \]
\end{cor}

\subsubsection*{Sidenote on formal topology}  

When evaluating a general simulation~$R: w_1^\ast \linear w_2$ directly (\ie
not relying on Corollary~\ref{cor:RTCbot_RTC_WTC}), we need to compute an
element of~$i_2 \IN \NU{w_2^\bot}$ from
\begin{itemize}
  \item
    a state~$i_1 : I_1$ and a state~$i_2 : I_2$,
  \item
    an element~$r : R(i_1,i_2)$,
  \item
    an element~$T_1 : i_1 \IN \NU{w_1^\bot}$.
\end{itemize}
We then need to find, for each branch~$a_2 : A_2(i_2)$, a
corresponding~$d_2 : D_2(i_2,a_2)$ and a way to continue the computation.
Given~$a_2$, by the general simulation property, we can compute an
element of
\[
  \SI{\alpha_1 : A_1^\ast(i_1)}
  \PI{\delta_1 : D_1^\ast(i_1,a_1)}
  \SI{d_2 : D_2(i_2,a_2)}
  R\big(i_1[\alpha_1/\delta_1] , i_2[a_2/d_2]\big)
\]
which can be rewritten as
\[
  i_1 \IN
  \Sem{w_1^\ast} \bigg(
  \LAM{i:I_1}
  \SI{d_2 : D_2(i_2,a_2)}
  R\big(i , i_2[a_2/d_2]\big)
  \bigg)
  \,.
\]
The crucial first step is matching the well-founded tree~$\alpha_1$ with the
infinite tree~$T_1$. This can be done with
\begin{lem}\label{lem:SambinExecution}
  Let~$I:\Set$, and~$w$ an indexed container over~$I$, and~$X$ a predicate
  over~$I$. The type~$w^\ast(X) \meets \NU{w^\bot} \to X \meets
  \NU{w^\bot}$ is inhabited.
\end{lem}
\begin{proof}\checked
  By matching dual quantifiers coming from~$i \IN w^\ast(X)$
  and~$i\IN\nu_{w^\bot}$, we get an alternating sequence of moves~$a_i$ and
  counter moves~$d_i$ reaching a final state~$i_f$ in~$X$, together with a
  infinite tree in~$i_f \IN \NU{w^\bot}$. More formally, suppose~$w^\ast(X)
  \meets \NU{w^\bot}$ is inhabited, i.e., that we have $i : I$, $\langle
  \alpha, f \rangle : i\IN w^\ast(X)$ and a (non-well-founded) tree $T$
  in~$i\IN\NU{w^\bot}$. We examine~$\alpha$:

 \begin{itemize}
   \item
     $\alpha=\leaf(i)$: we have $f\,\star : i\IN X$, in which
     case~$X\meets\NU{w^\bot}$ is inhabited.

   \item
     $\alpha=\node(a,k)$: where~$k:\PI{d} i[a/d]\IN w^\ast(X)$. We can
     apply~$\nuElim$ to ~$T:i \IN \NU{w^\bot}$ to obtain a function
     in~$\PI{a}\SI{d}i[a/d]\IN\NU{w^\bot}$. Applying that function
     to~$a:A(i)$, we get~$d:D(i,a)$ s.t.
     \begin{itemize}
       \item
         $i[a/d] \IN \NU{w^\bot}$,
       \item
         $k\,d : i[a/d] \IN w^\ast(X)$.
     \end{itemize}
     This pair inhabits~$w^\ast(X)\meets\NU{w^\bot}$, and by
     induction hypothesis, yields an inhabitant of~$X\meets\NU{w^\bot}$.
     \qedhere
 \end{itemize}
\end{proof}
This formula is at the heart of formal topology, where it is called
``monotonicity''~\cite{DBLP:journals/apal/CoquandSSV03}. There,~$i \IN
\Sem{w^\ast}(U)$ is read ``the basic open~$i$ is covered by~$U$'' (written~$i
\triangleleft U$), and~$i\IN\NU{w^\bot}$ is read ``the basic open~$i$ contains
a point'' and is written~$i\IN\mathsf{Pos}$.

Applied to the present situation with~$X= \LAM{i:I_1} \SI{d_2 :
D_2(i_2,a_2)} R\big(i , i_2[a_2/d_2]\big)$, we get an element
of~$X\meets\NU{w_2^\bot}$, which is precisely given by
\begin{itemize}
  \item
    a state~$i_1 : I_1$,
  \item
    a pair~$\langle d_2 , r\rangle$ with~$d_2:D_2(i_2,a_2)$ and~$r:R(i_1,
    i_2[a_2/d_2])$,
  \item
    an element~$T_1:\NU{w_1^\bot}$.
\end{itemize}
In other words, we get the sought after~$d_2$ (giving the first element of
the~$a_2$ branch), together with enough data to continue the computation.

\subsection{The Free Monad Construction as a Comonad}  

Composition of general simulations does not follow directly from composition of
linear simulations because there is a mismatch on the middle container:
composing~$R : w_1^\ast \linear w_2$ and~$S:w_2^\ast \linear w_3$ to
get a simulation from~$w_1^\ast$ to~$w_3$ is not obvious.
Fortunately, the operation~$w \mapsto w^\ast$ lifts to a comonad, and the
composition corresponds to composition in its (co)Kleisli category.

\begin{prop}
  Let~$\C$ be a locally cartesian closed category. The operation~$P \mapsto
  P^\ast$ lifts to a monad on the category of polynomial functors over~$I$
  with cartesian natural transformations between them.
\end{prop}
\begin{proof}\partlychecked
  The operation~$P\mapsto P^\ast$ goes from~$\End(\C/I)$, the category of
  polynomial endofunctors on~$\C/I$ to~$\Mnd(\C/I)$, the category of
  (polynomial) monads over~$\C/I$~\cite{GambHyl,polyMonads}.
  We write~$\Nat(F,G)$ for natural transformations from~$F$ to~$G$,
  and~$\Nat_{\Mnd}(F,G)$ for those transformations that respect the monad
  structures of~$F$ and~$G$.
  Writing~$\Alg{F}$ for the category of $F$-algebras, and~$\MAlg{F}$
  (when~$F$ is a monad) for the category of~$F$-algebras that respect the
  monad operations, we have
  \begin{myequation}
    \Nat_\Mnd(P^\ast, M)
    &\cong&
    \MAlg{M} \longrightarrow_\C \MAlg{P^\ast}
    &\text{\footnotesize\cite[proposition (5.3)]{Barr70}}
    \\
    &\cong&
    \MAlg{M} \longrightarrow_\C \Alg{P}
    &\text{\footnotesize\cite[proposition 17]{GambHyl}}
    \\
    &\cong&
    \Nat(P, M)
    &\text{\footnotesize\cite[proposition (5.2)]{Barr70}}
    \\
  \end{myequation}
  This shows that~$\BLANK^\ast$ is left adjoint to the forgetful
  functor~$\mathcal{U} : \Mnd(\C/I) \to \End(\C/I)$, and makes the
  composition~$\mathcal{U}\BLANK^\ast$ a monad. It only remains
  to show that the monad operations are strong cartesian transformations.

  Since strong natural transformations from~$\Sem{w_1}$ to~$\Sem{w_2}$
  correspond exactly to linear simulations~$(\equiv,\rho) : w_2 \linear
  w_1$~\cite{polyMonads,PolyFunctors}, it is enough to define the monad
  operations as simulations:
  \begin{itemize}
    \item
      $(\equiv,\varepsilon_w) : w^\ast \linear w$
    \item
      and~$(\equiv,\delta_w) : w^\ast \linear w^{\ast\ast}$
  \end{itemize}
  Those constructions are relatively straightforward in Agda: the first operation corresponds to
  embedding a single action~$a:A(i)$ into~$A^\ast(i)$ as~$\node(a,\LAM{d:D(i,a)}
  \leaf)$.
  A direct definition of~$\delta_w$ is done by induction and can be found in
  the Agda code. Semantically speaking, it can be derived from
  Lemma~\ref{lem:RTC_mu}:
  \begin{myequationqed}
    &&
    \Sem{w}\Sem{w^\ast}(X) \sub \Sem{w^\ast}(X)
    &\text{\footnotesize(first point of Lemma~\ref{lem:RTC_mu})}
    \\
    &\to&
    \Sem{w^\ast}(X) \cup \Sem{w}\Sem{w^\ast}(X) \sub \Sem{w^\ast}(X)
    &\text{\footnotesize(because $\Sem{w^\ast}(X)\sub\Sem{w^\ast}(X)$)}
    \\
    &\to&
    \Sem{w^\ast}\Sem{w^\ast}(X) \sub \Sem{w^\ast}(X)
    &\text{\footnotesize(second point of Lemma~\ref{lem:RTC_mu})}
    \\
    &\to&
    X \cup \Sem{w^\ast}\Sem{w^\ast}(X) \sub \Sem{w^\ast}(X)
    &\text{\footnotesize(because $X \sub \Sem{w^\ast}(X)$)}
    \\
    &\to&
    \Sem{w^{\ast\ast}}(X) \sub \Sem{w^\ast}(X)
    &\text{\footnotesize(second point of Lemma~\ref{lem:RTC_mu})}
  \end{myequationqed}
\end{proof}
What makes this lemma interesting is that the constructions themselves are
easy to define in Agda without identity types. A purely type theoretic proof
that the constructions satisfies the monad laws could not be completed in Agda,
because the overhead of reasoning with equality on dependent structures is
\emph{very} tedious. The categorical proof guarantees that it holds in all
models for extensional type theory which is good enough for us.

Because of the reversal (strong natural transformation from~$w_1$ to~$w_2$
are equivalent to identity linear simulations from~$w_2$ to~$w_1$), this
translates to ``$w\mapsto w^\ast$ lifts to a comonad in the category of
indexed containers with linear simulations''.
\begin{cor}\label{cor:comonad}
  The operation~$w \mapsto w^*$ lifts to a comonad in the category of
  indexed containers over~$I$ and linear simulations.
\end{cor}

\subsection{Composition of General Simulations}  

Recall that a comonad may be given in triple form with the following data:
\begin{itemize}
  \item
    a natural transformation~$\varepsilon_w : w^\ast \linear w$,

  \item
    a ``cobinding'' operation taking~$R : w_1^\ast \linear w_2$ to~$R^\sharp :
    w_1^\ast \linear w_2^\ast$ (defined as~$R^\ast \circ \delta_{w_1}$),

\end{itemize}
satisfying the following laws:
\begin{enumerate}
  \item $\varepsilon_w^\sharp = \id_{w^\ast}$,
  \item $\varepsilon_{w_2} \circ R^\sharp = R$ for~$R:w_1^\ast \linear w_2$,
  \item $(S \circ R^\sharp)^\sharp = S^\sharp \circ R^\sharp$ for~$R:w_1
    \linear w_2^\ast$ and~$S:w_2^\ast \linear w_3^\ast$.
\end{enumerate}
We can now define
\begin{defi}\label{def:general_comp}
    If~$R:w_1^\ast \linear w_2$ and~$S:w_2^\ast \linear w_3$, define~$S\bullet
    R$ with~$S\circ R^\sharp$.
\end{defi}
The comonad laws are then enough to prove that composition of general
simulations corresponds to composition of their evaluations.
\begin{prop}\label{prop:comp_general}
  Let~$w_1,w_2,w_3$ be containers indexed on~$I_1$, $I_2$ and~$I_3$,
  with~$R:w_1^\ast \linear w_2$ and~$S: w_2^\ast \linear w_3$,
  then we have
  \[
    (\eval_S^\ast\,i_2\,i_3\,s) \circ (\eval_R^\ast\,i_1\,i_2\,t)
    \quad\approx\quad
    \eval_{S\bullet R}^\ast\,i_1\,i_3\,⟨i_2,⟨s,t⟩⟩
  \]
  where
  \begin{itemize}
    \item
      $i_1:I_1$, $i_2:I_2$, $i_3:I_3$,

    \item
      $s:S(i_2,i_3)$ and~$t:T(i_2,i_3)$,

    \item
      and thus, $⟨i_2,⟨s,t⟩⟩: (T \bullet S)(i_1,i_3)$.
  \end{itemize}
\end{prop}
\begin{proof}\partlychecked
  Because~$\NU{w_2^{\ast\bot}}$ is a weakly terminal algebra
  for~$\Sem{w_2^\bot}$, by Lemma~\ref{lem:WTC}, we have a pair of morphisms
  $f_2 : i_2 \IN \NU{w_2^\bot} \to i_2 \IN \NU{w_2^{\ast\bot}}$ and $g_2 :
  i_2 \IN \NU{w_2^{\ast\bot}} \to i_2 \IN \NU{w_2^\bot}$ such that~$f_2g_2
  \approx \id$ and~$g_2f_2 \approx \id$.

  Expanding the definitions (and omitting all non-essential parameters, except
  for the first and last lines), we have:

  \begin{myequation}
    &&
    i_1 \IN \NU{w_1^\bot}
    \xrightarrow{ \eval^\ast_{S\bullet R}\,i_1\,i_3\,\langle i_2 \langle r,s\rangle\rangle}
    i_3 \IN \NU{w_3^\bot}
    \\
    &=&
    \NU{w_1^\bot}
    \xrightarrow{ \eval_{S\bullet R}^\ast}
    \NU{w_3^\bot}
    \\
    \text{\tiny(def of $\eval^\ast$)}
    &=&
    \NU{w_1^\bot}
    \xrightarrow{f_1}
    \NU{w_1^{\ast\bot}}
    \xrightarrow{ \eval_{S\bullet R}}
    \NU{w_3^\bot}
    \\
    \text{\tiny(def of~$\bullet$)}
    &=&
    \NU{w_1^\bot}
    \xrightarrow{f_1}
    \NU{w_1^{\ast\bot}}
    \xrightarrow{ \eval_{S\circ R^\sharp}}
    \NU{w_3^\bot}
    \\
    \text{\tiny(Lemma~\ref{lem:composition})}
    &=&
    \NU{w_1^\bot}
    \xrightarrow{f_1}
    \NU{w_1^{\ast\bot}}
    \xrightarrow{ \eval_{R^\sharp}}
    \NU{w_2^{\ast\bot}}
    \xrightarrow{ \eval_{S}}
    \NU{w_3^\bot}
    \\

    &=&
    \NU{w_1^\bot}
    \xrightarrow{f_1}
    \NU{w_1^{\ast\bot}}
    \xrightarrow{ \eval_{R^\sharp}}
    \NU{w_2^{\ast\bot}}
    \xrightarrow{ \id }
    \NU{w_2^{\ast\bot}}
    \xrightarrow{ \eval_S}
    \NU{w_3^\bot}
    \\

    \text{\tiny(remark above)}
    &\approx&
    \NU{w_1^\bot}
    \xrightarrow{f_1}
    \NU{w_1^{\ast\bot}}
    \xrightarrow{ \eval_{R^\sharp}}
    \NU{w_2^{\ast\bot}}
    \xrightarrow{g_2}
    \NU{w_2^{\bot}}
    \xrightarrow{f_2}
    \NU{w_2^{\ast\bot}}
    \xrightarrow{ \eval_S}
    \NU{w_3^\bot}
    \\

    \text{\tiny(remark below)}
    &=&
    \NU{w_1^\bot}
    \xrightarrow{f_1}
    \NU{w_1^{\ast\bot}}
    \xrightarrow{ \eval_{R^\sharp}}
    \NU{w_2^{\ast\bot}}
    \xrightarrow{ \eval_{\varepsilon_{w_2}}}
    \NU{w_2^{\bot}}
    \xrightarrow{f_2}
    \NU{w_2^{\ast\bot}}
    \xrightarrow{ \eval_S}
    \NU{w_3^\bot}
    &
    (*)
    \\

    \text{\tiny(comonad law)}
    &=&
    \NU{w_1^\bot}
    \xrightarrow{f_1}
    \NU{w_1^{\ast\bot}}
    \xrightarrow{ \eval_R}
    \NU{w_2^\bot}
    \xrightarrow{f_2}
    \NU{w_2^{\ast\bot}}
    \xrightarrow{ \eval_S}
    \NU{w_3^\bot}
    \\

    \text{\tiny(def of $\eval^\ast$)}
    &=&
    \NU{w_1^\bot}
    \xrightarrow{ \eval_R^\ast}
    \NU{w_2^\bot}
    \xrightarrow{ \eval_S^\ast}
    \NU{w_3^\bot}
    \\

    &=&
    i_1 \IN \NU{w_1^\bot}
    \xrightarrow{ \eval^\ast_{R}\,i_1\,i_2\,r}
    i_2 \IN \NU{w_2^\bot}
    \xrightarrow{ \eval^\ast_{S}\,i_2\,i_3\,s}
    _3 \IN \NU{w_3^\bot}
    \\
  \end{myequation}
  The only missing part~$(*)$ is showing that~$\eval_{\varepsilon_{w_2}} \approx g_2$
  where~$g_2 : i_2 \IN \NU {w_2^{\ast\bot}} \to i_2 \IN \NU{w_2^\bot}$ is
  the mediating morphism coming from the~$\Sem{w_2^\bot}$ coalgebra structure
  of~$i_2 \IN \NU{w_2^{\ast\bot}}$. This was formally proved in Agda. (The
  reason is essentially that both morphisms are, up-to bisimilarity, the
  identity.)
\end{proof}
Note that this proof is hybrid with a formal part proved in Agda, and a pen
and paper proof.



\section{Layering and Infinite Trees}

General simulations allow to represent all computable (continuous) functions
on streams: for a given piece of the output, we only need to look at
finitely many elements of the stream. A general simulation does that by
asking as many elements as it needs. However, for branching structures,
general simulations are not enough: general simulations explore their
arguments along a single branch $a_1/d_1, a_2/d_2, \dots$. For example, the
function summing each layer of a binary tree to form a stream is not
representable by a general simulation:
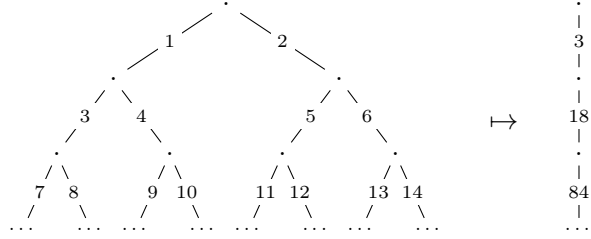
\[\begin{tikzpicture}[  
      level distance=1cm,
      level 1/.style={sibling distance=3cm},
      level 2/.style={sibling distance=1.5cm},
      level 3/.style={sibling distance=.9cm},
      tree node/.style={draw=none},
      every child node/.style={tree node},
      baseline={([yshift=-1ex] current bounding box.center)},
      ]
  \node[tree node] (Root) {.}
  child {
    node {.}
      child {node {.}
        child {node {\tiny\dots} edge from parent node[over] {\tiny 7}}
        child {node {\tiny\dots} edge from parent node[over] {\tiny 8}}
      edge from parent node[over] {\tiny 3}}
      child {node {.}
        child {node {\tiny\dots} edge from parent node[over] {\tiny 9}}
        child {node {\tiny\dots} edge from parent node[over] {\tiny 10}}
      edge from parent node[over] {\tiny 4}}
    edge from parent node[over] {\tiny1}}
  child {
    node {.}
      child {node {.}
        child {node {\tiny\dots} edge from parent node[over] {\tiny 11}}
        child {node {\tiny\dots} edge from parent node[over] {\tiny 12}}
      edge from parent node[over] {\tiny 5}}
      child {node {.}
        child {node {\tiny\dots} edge from parent node[over] {\tiny 13}}
        child {node {\tiny\dots} edge from parent node[over] {\tiny 14}}
      edge from parent node[over] {\tiny 6}}
    edge from parent node[over] {\tiny2}};
\end{tikzpicture}
\quad\mapsto\quad
\begin{tikzpicture}[
      level distance=1cm,
      level 1/.style={sibling distance=3cm},
      level 2/.style={sibling distance=1.5cm},
      level 3/.style={sibling distance=.9cm},
      tree node/.style={draw=none},
      every child node/.style={tree node},
      baseline={([yshift=-1ex] current bounding box.center)},
      ]
  \node[tree node] (Root) {.}
  child {
    node {.}
      child {node {.}
        child {node {\tiny\dots} edge from parent node[over] {\tiny 84}}
      edge from parent node[over] {\tiny 18}}
    edge from parent node[over] {\tiny 3}}
    ;
\end{tikzpicture}
\]
We want a notion of ``backtracking'' transducer that can explore
several branches. Describing such a transducer is difficult in type theory if
we try to refrain from using equality.


\subsection{Layering}  
\label{sub:layering}

To give a simulation access to several branches, we are going to replace a
branching structure by the stream of its ``layers''. Then, a simulation will
be able to access as many layers as it needs to get information about as many
branches as it needs.

\begin{defi}
  Given an indexed container~$w = ⟨ A,D,n⟩ $ over~$I:\Set$ we define an
  indexed container~$w^\sharp = \langle A^\sharp,D^\sharp,n^\sharp \rangle$ on~$I$
  by induction-recursion on~$\Fam(I)$.
  \begin{itemize}
    \item
      given~$i:I$, the index set~$A^\sharp(i):\Set$ is defined with
      \[
        \Rule{}{\leaf : A^\sharp(i)}
        \qquad
        \Rule{ \alpha : A^\sharp(i)
               \qquad
               l : \PI{\beta : D^\sharp(i,\alpha) } A(n^\sharp\,i\,\alpha \, \beta) }
             { (\alpha \triangleleft l) : A^\sharp(i) }
             \,;
      \]

    \item
      for~$\alpha : A^\sharp(i)$, the family $\big\langle
      D^\sharp(i,\alpha):\Set, n^\sharp\,i\,\alpha : D^\sharp(i,\alpha) \to
      I\big\rangle$ is defined with

      \begin{itemize}
        \item
          $D^\sharp(i,\leaf) = \One$ and
          $D^\sharp(i,\alpha \triangleleft l) = \SI{\beta :
          D^\sharp(i,\alpha)}
          D\big(n^\sharp\,i\,\alpha\,\beta,l\,\beta\big)$,

        \item
          $n^\sharp\,i\,\leaf\,\star = i$
          and
          $n^\sharp\,i\,(\alpha \triangleleft l)\,⟨ \beta,d⟩
          =
          n\ (n^\sharp\,i\,\alpha\,\beta)\ (l\,\beta)\ d$.
      \end{itemize}
  \end{itemize}
  The Agda definition looks like
  \begin{allttt}
    mutual
      data A# : I → Set where
        Leaf : \{ i : I \} → A# i
        \_◂\_ : \{i : I\} → (t : A# i) → ((b : D# i t) → A (n# t b)) → A# i

      D# : (i : I) → A# i → Set
      D# i Leaf = One
      D# i (t ◂ l) = Σ (D# i t) (λ ds → D (n# t ds) (l ds))

      n# : \{i : I\} → (t : A# i) → D# i t → I
      n# \{i\} Leaf * = i
      n# \{i\} (t ◂ l) ( ds , d ) = n \{n# \{i\} t ds\} (l ds) d
  \end{allttt}
\end{defi}
This definition is a direct extension of the operation appearing in previous
work about continuous functions and (unindexed)
containers~\cite{hancock09:_contin_funct_final_coalg}. It generalizes the
example of complete binary trees as an inductive recursive definition from
page~\pageref{ex:TreeIR} to arbitrary dependent~$A$-labelled and~$D$-branching
trees. Given such a tree~$\alpha : A^\sharp(i)$,~$D^\sharp(i,\alpha)$ indexes
the set of its terminal leaves. A new layer assigns a new element~$a:A(\dots)$
at each leaf, and~$\alpha \triangleleft l$ is the new, increased tree.

Of particular interest is the indexed container~$w^{\bot\sharp}$. An element
of~$A^{\bot\sharp}(i)$ is a complete tree of finite depth where branching occurs
at~$A(\BLANK)$ and labels are chosen in~$D(\BLANK,\BLANK)$. In particular, an
element of~$D^{\bot\sharp}(i,\alpha)$ is a finite sequence of actions.

\medbreak
We can now construct a lopsided (angelic) indexed container from an arbitrary
indexed container:
\begin{defi}
  Given~$w$ a container indexed on~$I$ and a fixed~$i:I$, we define a new
  container~$\ALayered{w,i}$:
  \begin{itemize}
    \item
      $\ALayered{w,i}$ is indexed on $A^\sharp(i)$,

    \item
      actions in state~$\alpha:A^\sharp(i)$ are given by~$\PI{\beta :
      D^\sharp(i,\alpha)}A(n^\sharp \, i \, \alpha \, \beta)$, i.e., ``layers'' on top
      of~$\alpha$,

    \item
      responses are trivial: there is only ever one possible
      response:~$\star$,

    \item
      the next state after action~$l$ in state~$\alpha$ (and response~$\star$)
      is~$\alpha\triangleleft l$.
  \end{itemize}
\end{defi}
States of this new container record a complete tree of finite depth.
Moreover, since it has trivial responses, an element of~$\leaf \IN
\NU{\ALayered{w,i}}$ is not very different from a (dependent) stream of
layers, and so, from an infinite tree in~$i \IN \NU{w}$. This is generalized
and formalized in the next lemmas.
\begin{lem}
  The predicate $\PI{\beta :D^\sharp(i,\alpha)} (n^\sharp\,i \, \alpha \,
  \beta) \IN \NU{w}$,\footnote{An element of~$\PI{\beta :D^\sharp(i,\alpha)}
  (n^\sharp\,i \, \alpha \, \beta) \IN \NU{w}$ is a way to extend the complete
  tree~$\alpha$ of finite depth to a full infinite complete tree by appending
  an infinite tree at each leaf.} depending on $\alpha : A^\sharp(i)$,  is a
  weakly terminal coalgebra for~$\Sem{\ALayered{w,i}}$.
\end{lem}
\begin{proof}\checked
  The proof corresponds to the above remark that an infinite tree is
  equivalently given by the infinite stream of its layers. It has been
  formalized in Agda.
\end{proof}
\begin{cor}\label{cor:layeringlemma}
  Given~$\alpha:A^\sharp(i)$, there are functions
  \[
    \PI{\beta :D^\sharp(i,\alpha)}
      (n^\sharp\,i \, \alpha \, \beta) \IN \NU{w}
    \quad\stackrel\approx{\longleftrightarrow}\quad
    \alpha \IN \NU{\ALayered{w,i}}
  \]
  that are, up to bisimulation, inverse to each other.
  In particular, there are functions
  \[
    f :
    i\IN\NU{w}
    \to
    \leaf \IN \NU{\ALayered{w,i}}
    \quad \text{and} \quad
    g :
    \leaf \IN \NU{\ALayered{w,i}}
    \to
    i\IN\NU{w}
  \]
  such that~$fg \approx \id$ and~$gf\approx\id$.
\end{cor}
\begin{proof}
  This is a direct consequence of Lemma~\ref{lem:WTC}.
\end{proof}

As mentioned on page~\pageref{rk:lopsided}, lopsided containers are special
because duality is essentially involutive, in the sense that~$w$
and~$w^{\bot\bot}$ are isomorphic:
\begin{itemize}
  \item
    there are bijections~$f_i$ between~$A(i)$ and~$A^{\bot\bot}(i)$,
  \item
    there is are bijections~$g_{i,a}$ between~$D(i,a)$ and~$D^{\bot\bot}(i,
    f_i\,a)$,
  \item
    the next state functions are compatible with them: $n(i,a,d) \equiv
    n^{\bot\bot}(i,f_i\,a,g_{i,a}\,d)$.
\end{itemize}
Rather than developing this notion, we will only state and prove the only
consequence we'll need.
\begin{lem}
  Suppose~$w$ has trivial (singleton) reactions, then $\NU{w}$ is a weakly
  terminal coalgebra for~$\Sem{w^{\bot\bot}}$.
\end{lem}
\begin{proof}\checked
\end{proof}
\begin{cor}\label{cor:duality_update}
  There are function
  \[
    \varphi :
    \leaf \IN\NU{\ALayered w,i}
    \to
    \leaf \IN \NU{\ALayered{w,i}^{\bot\bot}}
    \quad \text{and} \quad
    \psi :
    \leaf \IN \NU{\ALayered{w,i}^{\bot\bot}}
    \to
    \leaf \IN\NU{\ALayered w,i}
  \]
  such that~$\varphi\psi \approx \id$ and~$\psi\varphi\approx\id$.
\end{cor}

\subsection{Continuous Functions} 

We now have everything we need.
\begin{defi}
  A layered simulation from~$w_1$ to~$w_2$ at states~$i_1:I_1$ and~$i_2:I_2$
  is a simulation from~$\ALayered{w_1^\bot,i_1}^\bot$
  to~$\ALayered{w_2^\bot,i_2}^\bot$.
  A general layered simulation is a simulation
  from~$\ALayered{w_1^\bot,i_1}^{\bot\ast}$
  to~$\ALayered{w_2^\bot,i_2}^\bot$.
\end{defi}
\begin{thm}\label{lem:eval_layer}
  For every general layered simulation~$R:\ALayered{w_1^\bot,i_1}^{\bot\ast} \linear 
  \ALayered{w_2^\bot,i_2}^\bot$ there is an evaluation function
  \[
  \eval_R\,
  :
  R(\leaf, \leaf) \to
  i_1 \IN \NU{w_1^\bot} \to i_2 \IN \NU{w_2^\bot}
  \]
\end{thm}
\begin{proof}
  Given $r: R(\leaf,\leaf)$, we have
  \begin{myequationqed}
    i_1 \IN \NU{w_1^\bot}
    &\xrightarrow{\quad f_1\quad }&
    \leaf \IN \NU{\ALayered{w_1^\bot,i_1}}
    &\text{\footnotesize(Corollary~\ref{cor:layeringlemma})}
    \\
    &\xrightarrow{\ \varphi \ }&
    \leaf \IN \NU{\ALayered{w_2^\bot,i_2}^{\bot\bot}}
    &\text{\footnotesize(Corollary~\ref{cor:duality_update})}
    \\
    &\xrightarrow{\ \eval^\ast_R\,\leaf\,\leaf\,r\ }&
    \leaf \IN \NU{\ALayered{w_2^\bot,i_2}^{\bot\bot}}
    &\text{\footnotesize(Corollary~\ref{cor:eval*})}
    \\
    &\xrightarrow{\ \psi \ }&
    \leaf \IN \NU{\ALayered{w_2^\bot,i_2}}
    &\text{\footnotesize(Corollary~\ref{cor:duality_update})}
    \\
    &\xrightarrow{\quad g_2\quad }&
    i_2 \IN \NU{w_2^\bot}
    &\text{\footnotesize(Corollary~\ref{cor:layeringlemma})}
    \\
  \end{myequationqed}
\end{proof}
\begin{thm}\label{lem:eval_layer_comp}
  If composition of general layered simulations is defined as general
  composition of layered simulations, then evaluation of a composition is
  bisimilar to the composition of their evaluations.
\end{thm}
\begin{proof}
  The composition of evaluations gives
  \[
    i_1 \IN \NU{w_1^\bot}
    \xrightarrow{f_1}
    \BLANK
    \xrightarrow{\varphi_1}
    \BLANK
    \xrightarrow{\eval^\ast_R}
    \BLANK
    \xrightarrow{\psi_2}
    \BLANK
    \xrightarrow{g_2}
    \BLANK
    \xrightarrow{f_2}
    \BLANK
    \xrightarrow{\varphi_2}
    \BLANK
    \xrightarrow{\eval^\ast_S}
    \BLANK
    \xrightarrow{\psi_3}
    \BLANK
    \xrightarrow{g_3}
    i_3 \IN \NU{w_3^\bot}
    \,.
  \]
  Since~$f_2 g_2 \approx \id$ (Corollary~\ref{cor:layeringlemma})
  and~$\varphi_2\psi_2\approx\id$ (Corollary~\ref{cor:duality_update}), this whole
  composition is bisimilar to
  \[
    i_1 \IN \NU{w_1^\bot}
    \xrightarrow{f_1}
    \BLANK
    \xrightarrow{\varphi_1}
    \BLANK
    \xrightarrow{\eval^\ast_R}
    \BLANK
    \xrightarrow{\eval^\ast_S}
    \BLANK
    \xrightarrow{\psi_3}
    \BLANK
    \xrightarrow{g_3}
    i_3 \IN \NU{w_3^\bot}
  \]
  and thus, by Proposition~\ref{prop:comp_general}, to
  \[
    i_1 \IN \NU{w_1^\bot}
    \xrightarrow{f_1}
    \BLANK
    \xrightarrow{\varphi_1}
    \BLANK
    \xrightarrow{\eval^\ast_{S\bullet R}}
    \BLANK
    \xrightarrow{\psi_3}
    \BLANK
    \xrightarrow{g_3}
    i_3 \IN \NU{w_3^\bot}
    \,.
  \]
  This corresponds to evaluation of~$S \bullet
  R : \ALayered{w_1^\bot, i_1}^{\bot\ast} \linear \ALayered{w_2^\bot,
  i_2}^\bot$ as defined in Theorem~\ref{lem:eval_layer}.
\end{proof}



\section*{Concluding Remarks}

\subsection*{Internal Simulations}  
\label{sub:internal_sim}

It is possible to internalize the notion of linear simulation by defining an
indexed container~$w_1\linear w_2$ on~$I_1\times I_2$ satisfying~``$R$ is a
linear simulation from~$w_1$ to~$w_2$ iff~$R\sub \Sem{w_1\linear
w_2}(R)$''~\cite{hancock-apal06,polyDiagrams}.
\begin{defi}
  If~$w_1$ and $w_2$ are containers indexed on~$I_1$ and~$I_2$, then the
  container~$w_1 \linear w_2$, indexed on~$I_1 \times I_2$ is defined by

  \begin{itemize}
    \item
      $\Big.
      A\big(\langle i_1,i_2\rangle\big)
      =
      \SI{f:A_2(i_2) \to A_1(i_1)}\PI{a_2:A_2(i_2)} D_1\big(i_1,f(a_2)\big) \to D_2(i_2,a_2)$,

    \item
      $\Big.
      D\big(\langle i_1,i_2\rangle, \langle f, \varphi\rangle\big)
      =
      \SI{a_2:A_2(i_2)} D_1\big(i_1,{f(a_2)}\big)$,

    \item
      $\Big.
      n\,\langle i_1,i_2\rangle\,\langle f, \varphi\rangle\,\langle a_2, d_1\rangle
      =
      \big\langle i_1[f(a_2)/d_1], i_2[a_2/\varphi(a_2)(d_1)]\big\rangle$.
  \end{itemize}
\end{defi}
The resulting structure is nicer in the opposite category because
then,~$w_1\linear w_2$ generalizes duality (definition~\ref{def:duality}) in
the sense that~$w^\bot$ is the same as~$w \linear \bot$, where~$\bot$ is the
trivial container (indexed on~$I=\{\star\}$, with a single action and a
single reaction). Moreover, this definition is universal in the following
sense. Define the ``synchronous tensor'' $w_1\otimes w_2$ of containers, indexed with
the cartesian product of states with
\begin{itemize}
  \item 
    $\big(A_1\otimes A_2\big)\big(\langle i_1,i_2\rangle\big) = A_1(i_1)
    \times A_2(i_2)$,

  \item
    $\big(D_1 \otimes D_2\big)\big(\langle i_1, i_2\rangle, \langle
    a_1,a_2\rangle\big) = D_1(i_1,a_1) \times D_2(i_2,a_2)$,

  \item
    $\big(n_1 \otimes n_2\big)\ \langle i_1, i_2\rangle\ \langle
    a_1,a_2\rangle\ \langle d_1,d_2\rangle = \langle i_1[a_1/d_1],
    i_2[a_2/d_2]\rangle$.

\end{itemize}
One can show~\cite{polyDiagrams,PolyFunctors} that~$\linear$ is right-adjoint
to~$\otimes$. Linear simulations thus give rise to a symmetric monoidal
closed category.

That a simulation~$R:w_1 \linear w_2$  is nothing more than a coalgebra
for~$\Sem{w_1 \linear w_2}$ means that, up to bisimilarity, they can be seen
as elements of~$\NU{w_1\linear w_2}$. However, even if~$w_1$ and~$w_2$ are
finitary,~$w_1^\ast$ or~$\ALayered{w_1^\bot}^\ast$ are not and we cannot
iterate this construction to represent higher order continuous functions in
this way.

\subsection*{Thoughts about Completeness}   
\label{sub:completeness}

Stating and proving formally completeness of this representation is yet to be
done. Even though all functions definable\footnote{provided they pass the
termination criterion} in Agda are continuous (because they are computable),
we need to define continuity as an Agda predicate. Looking at the simplest
case of functions from streams to~\Two, we would need to transform some~$f :
\Stream{X} \to \Two$ into a well-founded tree with branching given by~$X$ and
leaves in~$\Two$. In order to do that, we need to know when a finite prefix of
a stream gives enough information to decide what it maps to, that is, to know
the ``modulus of continuity'' of the function, which is not definable in MLTT.
To avoid a contradiction (the modulus is trivially computable from the tree
representation), each continuous function needs to come with its modulus
function.

The definition of ``modulus of continuity'' for elements of
some~$\NU{w^\bot}(i)$ is not easy, and the simplest is probably to use the
exact same ideas that were developed: see a (finitary branching) coinductive
tree as the stream of its layer.

That's a preliminary step to be able to \emph{state} completeness of the
representation in Agda. This path does look neither very enlightening nor very
interesting.

\medbreak
On a meta-theoretical level, completeness of the representation is much
simpler:
\begin{enumerate}
  \item
    Semantically, when interpreting all the constructions in the category of
    sets and functions, every function that is continuous for the ``wild''
    topology from Section~\ref{sub:wild} is representable as a general layered
    simulation.\footnote{Let's stress the point again: proving such theorems
    in Agda is not possible, as there can be non-computable continuous
    functions.} This would be an analogous to
    Theorem~\ref{thm:stream_transducer} and the proof would go as follows
    \begin{itemize}
      \item
        we generalize Theorem~\ref{thm:stream_transducer} to dependent streams
        (\ie consider indexed containers with trivial actions) and show that
        in this case, any continuous function from~$\NU{w_1^\bot}(i_1)$
        to~$\NU{w_2^\bot}(i_2)$ is represented by an element of~$\NU{w_1^\ast
        \linear w_2}(i_1,i_2)$.

      \item
        we use Corollary~\ref{cor:duality_update} showing that
        any~$\NU{w^\bot}(i)$ (without restriction) is isomorphic to
        some~$\NU{w'^\bot}(j)$ where~$w'$ has trivial actions.

    \end{itemize}

  \item
    In particular, every function continuous for the natural topology
    (Section~\ref{sub:natural}) between greatest fixed points of finitary
    indexed containers\footnote{An indexed container is finitary if its
    sets of actions are finite.} is thus representable as a general layered
    simulation.

  \item
    If only the codomain is finitary, all continuous functions between fixed
    points of non-finitary containers are representable by simulations of the
    form~$\ALayered{w_1^\bot,i_1}^{\bot\ast} \linear w_2$. This is for example the
    case of the naturally continuous but non-wildly continuous function from
    page~\pageref{ex:natural_not_wild}. However, we do not know how to compose
    such simulations.

\end{enumerate}

\subsection*{Notes about Formalization}   

Some proofs have not been formalized in Agda, most notably:
\begin{itemize}
  \item
    the proof of Lemma~\ref{lem:bisim_id} or of Assumption~\ref{asm:bisim},

  \item
    the proof of Corollary~\ref{cor:comonad}.
\end{itemize}
We think the second holds in intensional type theory with function
extensionality but were unable to complete the proof. As it stands, we only
know it holds semantically by categorical reasoning. (It thus holds in
extensional type theory.)

The situation is subtle with Lemma~\ref{lem:bisim_id}. It is possible it can
be bypassed entirely. A direct proof of Corollary~\ref{cor:layeringlemma} was
in fact checked in Agda (with the \t{--with-K} flag), but its complexity
convinced us to base similar proofs on Lemma~\ref{lem:bisim_id}.

Enriching the type theory with stronger forms of equality may well simplify
the development. Preliminary investigation showed that \emph{cubical type
theory}~\cite{cubicalTT} allows to prove that bisimulation and (path) equality are identical.
For readers with some knowledge about cubical Agda~\cite{cubicalagda}, here is one way of
defining bisimulation:\footnote{Note that states are written~$s\in S$
to avoid clashing with cubical Agda's interval called~$I$.}
  \begin{allttt}
  record _≈_ \{s : S\} (T₁ T₂ : ν w s) : Type₀ where
    coinductive
    field
      root≈ : root T₁ ≡ root T₂
      branch≈ : (s' : S)
                (d₁ : D s (root T₁)) (q₁ : n s (root T₁) d₁ ≡ s')
                (d₂ : D s (root T₂)) (q₂ : n s (root T₂) d₂ ≡ s')
                (Pd : PathP (λ i → D s (root≈ i)) d₁ d₂) →
                (PathP (λ i → n s (root≈ i) (Pd i) ≡ s') q₁ q₂) →
                (subst (ν w) q₁ (branch T₁ d₁)) ≈ (subst (ν w) q₂ (branch T₂ d₂))
  \end{allttt}
This is much more streamlined than the definition in plain Agda!

It is then possible to prove the following:
\[
  \PI{s : S}\PI{T_1,T_2 : \NU{w}\ s}
  \quad
  \big(T_1 \approx T_2\big) \equiv \big(T_1 \equiv T_2\big)
\]
which makes all the bisimilarity proofs unnecessary.\footnote{To be honest,
the difficult part of this proof is due to Anders Mortberg and is part of
cubical Agda's standard library: \texttt{Cubical/Codata/M/Bisimilarity.agda}}

Recall however all this is only needed to \emph{prove properties} about the
objects that are defined in plain Martin-Löf type theory without identity.





\section*{Acknowledgements}

I really want to thank Peter Hancock for all the discussions that led to this
paper. His knowledge of type theory and inductive-recursive definitions,
together with his reluctance to give in to the temptation of equality are at
the bottom of this work.


\bibliographystyle{alpha}
\bibliography{eating}

\end{document}